\documentclass[12pt]{extarticle}

\usepackage[utf8]{inputenc}
\usepackage{caption}
\usepackage{stmaryrd}
\oddsidemargin \evensidemargin
\usepackage{floatrow}
\usepackage[dvipsnames]{xcolor}
\usepackage{amsmath, amssymb, amsfonts}
\usepackage{mathtools}
\usepackage{enumerate}
\usepackage[all]{xy}
\usepackage[utf8]{inputenc}
\usepackage{bbm}
\usepackage{mathrsfs}
\usepackage{tikz,pgfplots,tikz-3dplot}
\usepackage{circuitikz}
\usetikzlibrary{decorations.markings}

\renewenvironment{thebibliography}[1]{
  \begin{oldthebibliography}{#1}
    \setlength{\itemsep}{0.42em} 

    \setlength{\parskip}{0em}
}
{
  \end{oldthebibliography}
}

\usepackage[english]{babel}
\usepackage{alphabeta}

\usepackage[colorlinks=true, allcolors=blue,linkcolor=Blue, urlcolor=black,citecolor=MidnightBlue,backref=page]{hyperref}

\renewcommand{\th}{\theta}

\usepackage{amsthm}
\usepackage{cleveref}

\numberwithin{equation}{section}
\newtheorem{theorem}{Theorem}[section]
\newtheorem{lemma}[theorem]{Lemma}

\newtheorem{proposition}[theorem]{Proposition}

\theoremstyle{definition}

\newtheorem{definition}[theorem]{Definition}
\newtheorem{procedure}[theorem]{Procedure}

\newtheorem{examplex}[theorem]{Example}
\newenvironment{example}
{\pushQED{\qed}\examplex}
{\popQED\endexamplex} 

\newtheorem{remarkx}[theorem]{Remark}
\newenvironment{remark}
{\pushQED{\qed}\remarkx}
{\popQED\endremarkx} 

\newtheoremstyle{citing}
{}
{}
{\itshape}
{}
{\bfseries}
{\textbf{.}}
{.5em}
{\thmnote{#3}}
{\theoremstyle{citing}
}

\DeclareMathOperator{\init}{in}
\DeclareMathOperator{\ind}{ind}

\DeclareMathOperator{\rank}{rank}

\DeclareMathOperator{\Ann}{Ann}

\DeclareMathOperator{\Char}{Char}

\DeclareMathOperator{\Sol}{\mathcal{S}ol}

\DeclareMathOperator{\Sing}{Sing}
\DeclareMathOperator{\Start}{Start}

\DeclareMathOperator{\Li}{Li}
\DeclareMathOperator{\ord}{ord}

\newcommand{\NN}{\mathbb{N}}
\newcommand{\CC}{\mathbb{C}}
\newcommand{\RR}{\mathbb{R}}
\newcommand{\ZZ}{\mathbb{Z}}
\newcommand{\QQ}{\mathbb{Q}}

\newcommand{\KK}{\mathbb{K}}
\renewcommand{\d}{\mathrm{d}}
\newcommand{\dx}{\partial}
\newcommand{\dy}[1]{\partial_{y_{#1}}}
\newcommand{\thy}[1]{\theta_{y_{#1}}}
\newcommand{\pt}{\partial}

\usepackage{algorithm}
\usepackage{algpseudocode}
\usepackage{listings}

\pgfplotsset{compat=1.16}

\usepackage[breakable,skins]{tcolorbox}
\newtcolorbox{tbox}[1][]{%
    breakable,
    enhanced,
    colframe=black,
    coltitle=white,
    #1
}

\title{$D$-Module Techniques for Solving Differential Equations in the Context of Feynman Integrals}
\author{Johannes Henn, Elizabeth Pratt, Anna-Laura Sattelberger, Simone Zoia}
\date{}

\begin{document}

\maketitle \thispagestyle{empty}

\begin{abstract}
Feynman integrals are solutions to linear partial differential equations with polynomial coefficients. Using a triangle integral with general exponents as a case in point, we compare $D$-module methods to dedicated methods developed for solving differential equations appearing in the context of Feynman integrals, and provide a dictionary of the relevant concepts. In particular, we implement an algorithm due to Saito, Sturmfels, and Takayama to derive canonical series solutions of regular holonomic $D$-ideals, and compare them to asymptotic series derived by the respective Fuchsian systems.
\end{abstract}

\vfill

{\small
\noindent{\bf Authors' addresses:}

\medskip

\noindent Johannes Henn, MPP$^{\dagger}$ \hfill {\tt henn@mpp.mpg.de}  

\noindent Elizabeth Pratt, UC Berkeley$^{\diamond}$ \hfill {\tt epratt@berkeley.edu} 

\noindent Anna-Laura Sattelberger, MPI~MiS$^{\flat}$ and KTH$^{\sharp}$
\hfill {\tt anna-laura.sattelberger@mis.mpg.de}  

\noindent Simone Zoia, INFN and University of Turin$^{\circ}$, and CERN$^{\tiny \triangle}$ {\em (current)}
 \hfill{\tt simone.zoia@cern.ch}

\bigskip

\noindent $^{\dagger}$ Max Planck Institute for Physics, Boltzmannstra{\ss}e~8,
85748 Garching, Germany

\noindent $^{\diamond}$ Dept.\ of Mathematics, UC Berkeley, 970 Evans Hall \# 3840, Berkeley, CA 94720-3840 USA

\noindent $^{\flat}$ Max Planck Institute 
for Mathematics in the Sciences, Inselstra{\ss}e~22, 04103~Leipzig, Germany

\noindent $^{\sharp}\,$  \hspace{-.35mm}Department of Mathematics, KTH Royal Institute of Technology, 100~44~Stockholm, Sweden

\noindent $^{\circ}\,\,$Dipartimento di Fisica and Arnold-Regge Center, Università di Torino, and INFN, Sezione di\linebreak \hspace*{1.5mm} Torino, Via P.\ Giuria 1, 10125 Torino, Italy

\noindent \hspace*{-1.75mm} $^{\tiny \triangle}$\hspace*{.5mm}CERN, Theoretical Physics Department, 1211 Geneva 23, Switzerland
}

\newpage 

{\hypersetup{linkcolor=black}
\setcounter{tocdepth}{2}
\tableofcontents
}

\newpage
\section{Introduction}
Differential equations are ubiquitous in physics. Their solutions are often expressed in terms of special functions. Well-known examples are the hypergeometric function, or the Hermite polynomials that appear in the solution of the quantum mechanical hydrogen atom.
In perturbative quantum field theory, a crucial bottleneck is the evaluation of Feynman integrals. In particular, obtaining analytic or numerical results for higher-loop Feynman integrals is paramount for deriving more precise theory predictions for processes at the Large Hadron Collider. Because of this, substantial research efforts are dedicated to evaluating Feynman integrals. These integrals can be seen as solutions to linear partial differential equations (PDEs) with polynomial coefficients.
Analyzing these PDEs and extracting useful information from them is an important problem in this field. In this paper, we analyze this problem from two perspectives. The first one consists of $D$-module techniques from algebraic analysis in mathematics, and the second one consists of dedicated tools developed for solving the specific differential equations satisfied by Feynman integrals.

\medskip

In the first approach, one uses the fact that systems of linear PDEs generate ideals in the Weyl algebra, which is denoted by $D$. Powerful $D$-module techniques~\cite{SST00,SatStu19} allow us to extract precious information from the PDEs. For instance, the singular locus of a $D$-ideal $I$ describes where the solutions to the system of PDEs encoded by $I$ may have singularities. The notion of a $D$-ideal being regular singular is inherently linked to its solutions having moderate growth---which is expected for Feynman integrals. The holonomic rank of a holonomic \mbox{$D$-ideal} tells us about the number of linearly independently solutions. It corresponds to the number of master integrals in the physics terminology. Finally, Gr\"obner basis methods allow us to compute canonical series solutions. These are solutions of a very special form: they are polynomials in the variables $x_i$'s and $\log(x_i)$'s, with coefficients in formal power series in the $x_i$. In the regular holonomic case, the solution space is fully spanned by solutions of that particular form. We compute these solutions with the help of algebro-geometric properties and the Gr\"{o}bner data of the $D$-ideal via an algorithm of Saito, Sturmfels, and Takayama~(SST). 
Given a regular holonomic \mbox{$D$-ideal} $I$
and a real weight vector $w$, this algorithm allows us to compute all terms up to a specified $w$-weight $k$ in the canonical series solutions to~$I$.
These truncated series take the~form
\begin{align}\label{eq:Fkintro}
F_k(x) \,=\ x^A \cdot \hspace*{-5mm} \sum _{\substack{0\leq p\cdot w\leq k,\, p\in C^*_{\mathbb{Z}}, \\ 0 \leq b_j < \rank(I)}} c_{pb} \, x^p\log(x)^b \, ,
\end{align}
where we use the multi-index notation $x^a = x_1^{a_1} \cdots x_n^{a_n}$, and similarly for the logarithms. 
The index set $C^*_{\mathbb{Z}}$ over which the sum in \eqref{eq:Fkintro} runs can be read completely from the $D$-ideal; this will be made precise in the article.

The second approach builds on various techniques that have been developed for evaluating Feynman integrals. The latter satisfy Picard--Fuchs equations~\cite{LV22,Muller-Stach:2012tgj}, which are higher-order linear PDEs for individual Feynman integrals, or, equivalently, systems of first-order equations for 
master integrals, see \cite{Argeri:2007up,Henn:2014qga} for reviews. The Fuchsian nature of the singularities---which is expected for Feynman integrals---allows one to derive asymptotic expansions around singular points using a method of Wasow~\cite{Wasow}.

\medskip 

There have been various interactions between the relevant mathematics and physics communities already, resulting in several interesting works, for example on determining the number of master integrals \cite{AFST22,FeynmanAnnihilator,Lee:2013hzt}, on deriving Picard--Fuchs equations for Feynman integrals \cite{LV22,Muller-Stach:2012tgj}, and on relating Feynman integrals and GKZ systems \cite{Ananthanarayan:2022ntm,MacaulayFeynman,Cruz19,MellinBarnes,KK1976,Walther22}. 
However, to our knowledge, a systematic synthesis and comparison of \mbox{$D$-module} techniques and methods employed in high energy physics has not yet been done. In this work, we initiate such a comparison, with the aim of providing insights for both communities. For this initial study, we focus on a particular system of physically motivated PDEs
and show how one obtains canonical series solutions of the form~\eqref{eq:Fkintro}, both using the SST algorithm and Wasow's method. In the latter, we capture the  weight vector from the Gr\"{o}bner techniques via an auxiliary variable.

\medskip

\textbf{Outline.} 
\Cref{sec:confdiffeq} introduces the systems of PDEs we are investigating here, together with their origin in physics. \Cref{sec:Weylalgebra} 
provides background on ideals in the Weyl algebra and their (multi-valued) holomorphic solutions that will be needed throughout this article. It also outlines how to compute canonical series solutions of regular holonomic \mbox{$D$-ideals} via the SST algorithm, which is based on Gr\"{o}bner basis computations in the Weyl algebra. 
In \Cref{sec:solI3Groebner}, we compute canonical series solutions for the three-particle case using the SST algorithm. In \Cref{sec:solWasow}, we recover these solutions via a method of Wasow for computing solutions of systems that are in Fuchsian form. We conclude with an outlook to future work in \Cref{sec:outlook}. Appendices \ref{appendix:conformal} and \ref{appendix:IntegrPfaf} contain some background on conformal symmetry as well as on integrable~connections. {In  Appendix~\ref{sec:fourloop}, we provide an application of this method to a four-loop Feynman integral.}

\section[Running example: a triangle integral \& conformal differential equations]{Running example: a triangle integral and conformal differential equations}\label{sec:confdiffeq}
As a test case for this work, we study a particular class of Feynman integrals: the one-loop ``triangle'' Feynman integrals,
\begin{align} \label{eq:triangle_mom}
J^{\text{triangle}}_{d;\nu_1,\nu_2,\nu_3} \,=\, \int_{\RR^d} \frac{\mathrm{d}^d k}{\mathrm{i} \pi^{\frac{d}{2}}} \frac{1}{
(-|k|^2)^{\nu_1} \, 
(-|k+p_1|^2)^{\nu_2} \, (-|k+p_1+p_2|^2)^{\nu_3} } \,,
\end{align}
which correspond to the Feynman graph shown in \Cref{fig:triangle}. 
These integrals are relevant to describe the scattering of three particles with momenta $p_1,p_2,p_3 \in \mathbb{R}^d$ (with $p_3=-p_1-p_2$ due to momentum conservation) in a $d$-dimensional Minkowski spacetime. We denote by $|v|^2$ the Minkowski norm of the $\mathbb{R}^d$-vector $v$, i.e., $|v|^2 \coloneqq v^{\top} \cdot g \cdot v$, where $g = \operatorname{diag}(1,-1,\ldots,-1)$ is the metric tensor of Minkowski spacetime. 
The triangle Feynman integrals depend on the momenta of the scattering particles only through three independent variables $x=\{x_1,x_2,x_3\}$,
\begin{align}
x_1 \,=\, |p_1|^2 \,, \qquad x_2 \,=\, |p_2|^2 \,, \qquad x_3 \,=\, |p_1+p_2|^2 \,.
\end{align}
This becomes apparent in the Feynman parameterization,
\begin{align}
J^{\text{triangle}}_{d;\nu_1,\nu_2,\nu_3} \ =\ \Gamma\left(\nu-\frac{d}{2}\right) \int_{\alpha_i \ge 0}  
\left(\prod_{i=1}^3 \mathrm{d}\alpha_i \,  \frac{\alpha_i^{\nu_i-1}}{\Gamma(\nu_i)} \right)  \delta\left(\alpha_1 + \alpha_2 + \alpha_3 -1 \right) \, \frac{\mathcal{U}^{\nu-d}}{\mathcal{F}^{\nu-\frac{d}{2}}} \,,
\end{align}
where $\nu=\nu_1+\nu_2+\nu_3$, $\delta$ is the Dirac delta function,
\begin{align}
\mathcal{U} \,=\, \alpha_1+\alpha_2+\alpha_3  \,,
\qquad \text{and} \qquad \mathcal{F} \,=\, -\alpha_1 \alpha_2 x_1 - \alpha_2 \alpha_3 x_2 - \alpha_3 \alpha_1 x_3 \,.
\end{align}
The triangle integrals with integer exponents $\nu_i$ were computed analytically in~\cite{Boos:1987bg,Davydychev:1992xr,Usyukina:1992jd,Usyukina:1992wz,Usyukina:1994iw}. 
Expressions valid for generic values of the exponents $\nu_i$ have been obtained in terms of Appell’s hypergeometric function $F_4$~\cite{Anastasiou:1999ui,Boos:1990rg,Coriano:2013jba,Cruz19} or triple-$K$ integrals~\cite{BMSconf}.

We approach the computation of these integrals by exploiting symmetries, which are encoded via differential equations. The triangle integrals are in fact solutions to a system of linear second-order~PDEs,
\begin{align} \label{eq:Pf}
    P_i \bullet f(x) \,=\, 0 \, \qquad \text{for} \ i=1,2,3\,,
\end{align}
originating from conformal symmetry, where
\begin{align}\label{I3system}
\begin{split}
P_1 &\,=\, 4(x_1\partial_1^2 - x_3\partial_3^2) + 2(2 + c_0 - 2c_1)\partial_1 -2(2 + c_0 - 2c_3)\partial_3 \, ,\\
P_2  & \,=\, 4(x_2\partial_2^2 - x_3\partial_3^2) + 2(2 + c_0 - 2c_2)\partial_2 -2(2 + c_0 - 2c_3)\partial_3 \, , \\
P_3 & \,=\, (2c_0 - c_1 - c_2 - c_3) + 2(x_3\partial_3 + x_2\partial_2 + x_1\partial_1)\, .
\end{split}
\end{align}
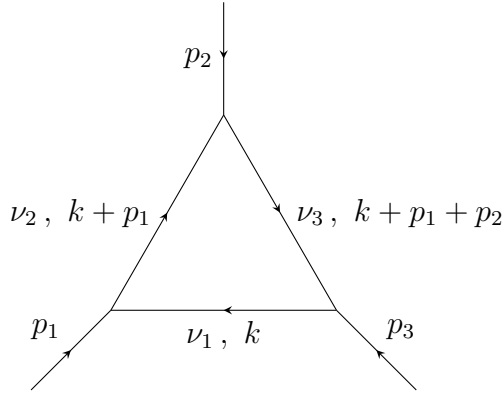
\begin{figure}[t!]
\begin{tikzpicture}[decoration={
    markings,
    mark=at position 0.5 with {\arrow{stealth}}}
    ] 
\coordinate (v1) at (0,0);
\coordinate (v2) at (1.5,2.598);
\coordinate (v3) at (3,0);
\coordinate (p1) at (-1.0607, -1.0607);
\coordinate (p2) at (1.5, 4.098);
\coordinate (p3) at (4.0607, -1.0607);
\draw[postaction={decorate}] (v1) -- (v2) node[left,midway,anchor=east] {$\nu_2 \, , \ k+p_1 \,$};
\draw[postaction={decorate}] (v2) -- (v3) node[right,midway,anchor=west] {$\, \nu_3 \,, \ k+p_1+p_2$};
\draw[postaction={decorate}] (v3) -- (v1) node[below,midway] {$\nu_1 \,, \ k$};
\draw[postaction={decorate}] (p1) -- (v1) node[above left,midway,anchor=south east] {$p_1$};
\draw[postaction={decorate}] (p2) -- (v2) node[left,midway,anchor=east] {$p_2$};
\draw[postaction={decorate}] (p3) -- (v3) node[above right,midway,anchor=south west] {$p_3$};
\end{tikzpicture}
\caption{The Feynman graph 
representing the one-loop triangle Feynman integrals defined in \Cref{eq:triangle_mom}. Due to momentum conservation, $p_1+p_2+p_3=0$. Next to the internal edges, we record the corresponding exponent $\nu_i$ as well as the loop momentum.}
\label{fig:triangle}
\end{figure}

\noindent Here, $c_0=d$ is the number of spacetime dimensions, and $c_1,c_2,c_3$ are the conformal weights. 
We give some background on conformal symmetry and derive these PDEs in \Cref{appendix:conformal}.~For 
\begin{align} \label{eq:c2nu}
  c_0 \,=\, d\,, \qquad c_1 \,=\, d - \nu_2 - \nu_3 \,, \qquad c_2 \,=\, d - \nu_1 - \nu_3 \,, \qquad c_3 \,=\, d - \nu_1 - \nu_2 \,,
\end{align}
the triangle integrals $J^{\text{triangle}}_{d;\nu_1,\nu_2,\nu_3}$ are solutions to the conformal PDEs~\eqref{eq:Pf}.

Let us review a few useful properties of the operators $P_i$ in~\eqref{I3system}. First of all, the operator~$P_3$ implies that the solutions are homogeneous of degree $(2 c_0-c_1-c_2-c_3)/2$. 
Secondly, the system~\eqref{eq:Pf} is symmetric under permutations of the variables $x_i$ (together with the corresponding conformal weight $c_i$). While the operator $P_3$ is manifestly symmetric, the permutations map $P_1$ and $P_2$ into $\mathbb{Q}$-linear combinations of themselves. As a result, the solution space is symmetric as well.
For computing the solutions, we will focus on the physically motivated case $c_0=4$, $c_1 = c_2 = c_3=2$, which corresponds to the (classically) conformal $\phi^4$-theory in four spacetime dimensions.\footnote{Other choices can also be of physical interest, and may be treated similarly. See e.g.~\cite{Coriano:2013jba} for a solution for general values of the indices in terms of Appell's hypergeometric function $F_4$. See also~\cite{Barnes:2010jp,Bautista:2019qxj,Bzowski:2020lip,BMSconf,Bzowski:2015yxv,Coriano:2013jba,Gillioz:2019lgs,Isono:2019ihz} for related work.}
In this case, the operators from~\eqref{I3system}~become 
\begin{align}\label{eq:p1p2p3tilde}
\begin{split}
\tilde{P}_1 &\,=\, x_1\partial_1^2 - x_3\partial_3^2 + \partial_1 -\partial_3 \,,\\
\tilde{P}_2 &\,=\,x_2\partial_2^2 - x_3\partial_3^2 + \partial_2 -\partial_3 \,,\\
\tilde{P}_3 &\,=\,  x_1\partial_1 + x_2\partial_2 + x_3\partial_3 + 1 \,.
\end{split}
\end{align}
Later on, we will consider the $D_3$-ideal generated by the operators in \eqref{eq:p1p2p3tilde} and will denote it by $I_3.$
Through \Cref{eq:c2nu}, we see that the relevant triangle integral has $d=4$ and unit exponents, $\nu_1=\nu_2=\nu_3=1$. A convenient analytic expression for the latter is
\begin{align} \label{eq:triangle}
\begin{split}
J^{\text{triangle}}_{4;1,1,1}(x_1,x_2,x_3) & \,=\, \frac{1}{\sqrt{\lambda}} 
\Big[ 2 \, \Li_{2}\left( \tau_2 \right) + 2 \, {\rm Li}_{2}\left( \tau_3 \right) + \frac{\pi^2}{3}  \\
& \qquad + \log \left( \frac{\tau_3}{\tau_2} \right) \log \left( \frac{1-\tau_3}{1-\tau_2} \right)  
+ \log \left(-{\tau_2} \right) \log \left(-\tau_3 \right) \Big] \,,
\end{split}
\end{align}
from~\cite{Zoia:2021zmb}, where the function $\Li_2$ denotes the {\em dilogarithm} (cf.~\cite{higherlog}),  $\lambda$ is the polynomial
\begin{align} \label{eq:lambda}
\lambda \,\coloneqq \,  x_1^2 + x_2^2 + x_3^2 - 2(x_1x_2 + x_1x_3 + x_2x_3) \,\in \, \CC[x_1,x_2,x_3] \,,
\end{align}
and
\begin{align} \label{eq:tau}
\tau_2 \,=\, - \, \frac{2 x_2}{x_1-x_2-x_3-\sqrt{\lambda}}  \quad \text{and} \quad \tau_3 \,=\,- \, \frac{2 x_3}{x_1-x_2-x_3-\sqrt{\lambda}} \, .
\end{align}
The function in \Cref{eq:triangle} is a solution of $I_3$. It is a multivalued function and its analytic continuations are solutions of $I_3$ as well. The $\CC$-vector space of the analytic continuations of $J^{\text{triangle}}_{4;1,1,1}$ is spanned by the four linearly independent functions
\begin{align}\label{eq:f1234}
\begin{split}
 f_1(x_1,x_2,x_3) & \,=\, J^{\text{triangle}}_{4;1,1,1}(x_1,x_2,x_3) \,, \\
f_2(x_1,x_2,x_3) & \,=\, \frac{1}{\sqrt{\lambda}} \log\left( \frac{x_1-x_2-x_3 - \sqrt{\lambda}}{x_1-x_2-x_3 + \sqrt{\lambda}} \right) , \\ 
f_3(x_1,x_2,x_3) & \,=\, \frac{1}{\sqrt{\lambda}} \log\left( \frac{x_2-x_1-x_3 - \sqrt{\lambda}}{x_2-x_1-x_3 + \sqrt{\lambda}} \right) , \\
f_4(x_1,x_2,x_3) & \,=\, \frac{1}{\sqrt{\lambda}} \,.
\end{split}
\end{align}
The functions $f_2,\, f_3, \, f_4$ are called {\em discontinuities} of $f_1 $ in physics, i.e., they are differences of $f_1$ and an analytic continuation of $f_1.$
This can be shown conveniently using the {\em symbol method}~\cite{Goncharov:2010jf} (for a review, see e.g.~\cite{Duhr:2014woa} and the  references therein). 
We will see in~\Cref{sec:solI3Groebner} that the holonomic rank of $I_3$ is $4$ and that this implies that the four functions in~\eqref{eq:f1234} span the whole space of holomorphic solutions to the system of PDEs encoded by $I_3$. All these functions have moderate growth when approaching points of the singular locus
\begin{align}\label{eq:SingIfs}
\left\{x_1 =0 \right\} \,\cup\, \left\{x_2 =0\right\} \,\cup\, \left\{x_3=0\right\} \,\cup\, \left\{ \lambda = 0 \right\},
\end{align}
and we will argue later that the $D$-ideal $I_3$ is indeed regular singular. As anticipated, the solution space is symmetric under permutations of the variables $x$: the functions $f_1$ and~$f_4$ are invariant, while $f_2$ and $f_3$ transform into $\mathbb{Q}$-linear combinations of themselves. In \Cref{sec:Fuchsian}, we will recover these solutions by solving the Pfaffian system associated with $I_3$.

\section{\texorpdfstring{$D$}{D}-ideals and canonical series solutions}\label{sec:Weylalgebra}
We here recall basics about ideals in the Weyl algebra and their (multi-valued) holomorphic solutions that will be needed for the computation of canonical series solutions later on. This theory will also allow us to derive crucial information about our system of PDEs, such as the singular locus and number of holomorphic solutions, before even computing the solutions. For further details about the theory of $D$-modules and holonomic functions, we refer our readers to \cite{HTT08,SST00,SatStu19} and the references~therein.

\subsection{\texorpdfstring{$D$}{D}-ideals and their solutions}\label{sec:Didealsols}
The ($n$-th) {\em Weyl algebra}, denoted $D_n$ or just $D$, is the free $\CC$-algebra 
\begin{align}
D \, \coloneqq \, \CC[x_1,\ldots,x_n]\langle \pt_1,\ldots,\pt_n \rangle
\end{align}
generated by $x_1,\ldots,x_n,\partial_1,\ldots,\partial_n$ modulo the following relations: all generators are assumed to commute, except $\partial_i$ and $x_i$.
Their commutator is 
\begin{align}
[\partial_i,x_i] \,=\, \pt_i x_i -x_i \pt_i \,=\, 1 \, .
\end{align}
Elements of $D$ are linear differential operators with polynomial coefficients:
\begin{align}
D \,=\, \left\{  \sum_{k \in \NN^n} a_k \partial_1^{k_1}\cdots \pt_n^{k_n} \, | \, a_k \in \CC[x_1,\ldots,x_n], \text{ only finitely many }a_k \text{ non-zero}  \right\}  \, .
\end{align}
The {\em rational Weyl algebra} $R_n=\CC(x_1,\ldots,x_n)\langle \pt_1,\ldots,\pt_n\rangle$ is the ring of linear differential operators with coefficients in the field of rational functions $\CC(x_1,\ldots,x_n) = \{ p/q \, | \, p,q \in \CC[x_1,\ldots,x_n], \, q\neq 0 \}$. We will denote the action of a differential operator $P$ on a function $f(x)$ by the symbol~$\bullet$. For instance, $\partial\bullet f$ denotes ${\partial f}/{\partial x}$.
\begin{example}
The Airy equation 
\begin{align}
f''(x) - xf(x) \,=\, 0
\end{align}
is encoded by the operator $P=\partial^2 -x\in D$. The Airy functions $\operatorname{Ai}$ and $\operatorname{Bi}$ are solutions of~$P$, i.e., $P\bullet \operatorname{Ai}=P \bullet \operatorname{Bi}=0$. They span the $2$-dimensional solution space of $P$.
\end{example}
A {\em (left) $D$-ideal} is a subset $I\subset D$ that is closed under addition and multiplication by elements of~$D$ from the left. They encode systems of linear PDEs. For a vector $(u,v)\in\RR^{2n}$, the $(u,v)$-{\em weight} of $x^{\alpha}\partial^{\beta}$, $\alpha,\beta \in \NN^n$, is $u\cdot \alpha + v\cdot \beta$. By $\operatorname{ord}_{(u,v)}(P)$, we denote the largest $(u,v)$-weight among the monomials appearing in $P$. The {\em characteristic ideal}
of a $D$-ideal $I$ is the ideal $\init_{(0,1)}(I)\subset \CC[x_1,\ldots,x_n][\xi_1,\ldots,x_n]$ generated by the parts of highest $(0,1)$-weight of all differential operators in $I$, where $(0,1)$ denotes the vector $(0,\ldots,0,1,\ldots,1)\in \RR^{2n}$, i.e., one puts weight $0$ to all $x_i$ and weight $1$ to all $\pt_i$. The characteristic ideal lives in the polynomial ring, where one replaces $\partial_i$ by $\xi_i$ to stress that they are now commuting variables. This is due to the fact that the commutator of two operators $P,Q$ has $(0,1)$-weight strictly smaller than $\deg(P)+\deg(Q)$. For $n=1$, the initial $\init_{(0,1)}(P)$ is precisely the principal symbol of the differential operator $P$. The {\em characteristic variety} of $I$ is 
\begin{align}
\Char(I) \, \coloneqq \, V(\init_{(0,1)}(I)) = \left\{ (x,\xi) \,|\, p(x,\xi) = 0 \text{ for all } p\in \init_{(0,1)}(I) \right\} \, \subset \, \CC^{2n} \, .
\end{align}
Here, $\CC^{2n}$ is the affine $2n$-space $\CC_x^n\times \CC_{\xi}^n$ with coordinates $x_1,\ldots,x_n,\, \xi_1,\ldots,\xi_n$. A $D$-ideal is {\em holonomic} if $\dim(\Char(I))=n$. By Bernstein's inequality, all components of $\Char(I)$ have dimension at least $n$; hence, a $D$-ideal $I$ is holonomic if the dimension of its characteristic variety is smallest possible. The Zariski closure of the projection of \mbox{$\Char(I)\setminus \{\xi_1 = \cdots = \xi_n = 0 \}$} to the $x$-coordinates is the {\em singular locus} of $I$ and is denoted by $\Sing(I) \subset \CC_x^n$. Algebraically, $\Sing(I)$ is computed as the elimination ideal
\begin{align}
\Sing(I) \, = \, \left( \init_{(0,1)}(I) \, \colon \langle \xi_1,\ldots,\xi_n \rangle^{(\infty)}\right) \, \cap \, \CC[x_1,\ldots,x_n] \, ,
\end{align}
where $(I\colon J^{(\infty)})$ denotes the saturation of an ideal $I$ by an ideal $J$, the definition of which we recall now. Let $I,J$ be ideals in  a polynomial ring $ \CC[x_1,\ldots,x_n]$. For $k\in \NN$, the {\em ideal quotient} is $(I\colon J^k) \coloneqq \{ p\in \CC[x_1,\ldots,x_n ] \, | \,  pJ^k \subset I\}$. Then, the {\em saturated ideal} $(I\colon J^\infty)$ with respect to $J$ is the $\CC[x_1,\ldots,x_n]$-ideal
\begin{align}
\left(I\colon J^\infty\right) \ \coloneqq \ \bigcup_{k\geq 1} \left(I\colon J^k\right).
\end{align}

\begin{definition}\label{def:holrank}
The {\em holonomic rank} of a $D$-ideal $I$ is 
\begin{align}
\rank(I) \coloneqq \dim_{\CC(x_1,\ldots,x_n)} \left( R_n/R_nI \right) \, .
\end{align}
\end{definition}
If $I$ is holonomic, it follows that $\rank(I)<\infty$. The reverse implication is not true. The holonomic rank can be computed as the number of standard monomials for a Gr\"{o}bner basis~of~$I$ in the rational Weyl algebra, see \Cref{sec:intPfaffian} for more details.

\begin{theorem}[Cauchy--Kovalevskaya--Kashiwara]\label{thm:CKK}
Let $I$ be a holonomic $D_n$-ideal.
On a simply connected domain $U\subset \CC^n\setminus \Sing(I)$, the $\CC$-vector space of holomorphic solutions to~$I$ on~$U$ has dimension~$\rank(I)$.
\end{theorem}

\begin{remark}
In \Cref{thm:CKK}, finite holonomic rank is sufficient for the statement to hold. One arrives at a holonomic $D$-ideal by taking the Weyl closure (see \Cref{def:WeylClosure}).
\end{remark}

A function $f$ is called {\em holonomic} if its annihilator
\begin{align}
    \Ann_{D}(f) \, \coloneqq \, \left\{ P\in D \, | \, P\bullet f = 0\right\} \, 
\end{align}
is a holonomic $D$-ideal. Numerous functions in the sciences are holonomic, e.g., hypergeometric functions, many trigonometric functions, some probability distributions, and many special functions 
such as Airy's or Bessel's functions, and polylogarithms. Zeilberger~\cite{Zei90} was the first to study holonomic functions in an algorithmic way. 

If all solutions to a $D_1$-ideal $I$ have moderate growth (cf.~\cite[p.~146]{vdPS}) when approaching the singular locus of~$I$, including at $\infty$ (and at boundary components of a compactification, respectively), the $D_1$-ideal is called {\em regular (singular)}, and {\em irregular singular} otherwise. For~$n>1$, it is more involved to read if a $D_n$-ideal is regular singular. 

Given an annihilating $D$-ideal of some functions, there is a straightforward way to slightly enlarge it, which is at the same time of theoretical interest: the Weyl closure.

\begin{definition}\label{def:WeylClosure}
The {\em Weyl closure} of a $D_n$-ideal $I$, denoted $W(I)$ is the $D_n$-ideal
\begin{align}
W(I) \, \coloneqq \, R_n I \cap D_n \, .
\end{align}
\end{definition} 
Clearly, $I\subseteq W(I)$. Hence, for the singular locus and the space of holomorphic solutions to the system of PDEs encoded by $I$, one has 
\begin{align}
\Sing(I) \,\supseteq\, \Sing(W(I)) \quad \text{and} \quad \Sol(I) \,\supseteq\, \Sol(W(I)) \, .
\end{align}  
Moreover, $\rank(I)=\rank(W(I))$, since $R_n I = R_nW(I)$. Since every element $Q$ of $W(I)$ can be written as $Q=r\cdot P$ for some $r\in \CC(x_1,\ldots,x_n)$ and $P\in I$, we also have the inclusion $\Sol(I)\subseteq \Sol(W(I)).$ Hence, $\Sol(I)=\Sol(W(I))$. In summary, a $D$-ideal $I$ and its Weyl closure $W(I)$ have the same solution space, but $W(I)$ might contain additional operators that annihilate all solutions of $I$.
If $\rank(I)< \infty$, it does not follow that $I$ itself is holonomic; but there is a rescue: taking the Weyl closure turns $I$ into a holonomic $D$-ideal. Hence, techniques that are valid for holonomic $D$-ideals can be applied to non-holonomic $D$-ideals of finite holonomic rank by passing to their Weyl closure.

\subsection{Pfaffian systems}\label{sec:intPfaffian}
If $I$ is a holonomic $D$-ideal, the $D$-module $D/I$ gives rise to a vector bundle of rank $\rank(I)$, say $\rank(I)=m$, with an integrable connection induced by the action of~$D$ on $D/I$. We refer the interested readers to \Cref{appendix:IntegrPfaf}, where we explain this geometric interpretation. In particular, the origin of the integrability conditions of Pfaffian systems becomes visible there. Here, we stick to what is needed for the sake of our computations in~this~paper. 

Let $S= \{ s_1,s_2,\ldots,s_m \}$ be the set of {\em standard monomials} of $I$ in $R_n$ for a chosen term ordering $\prec$, i.e., those monomials $ \partial^b$, $b \in \NN^n$, that are not contained in the initial ideal of~$I$; cf.~\cite[p.~33]{SST00} for more details. Without loss of generality, we can assume $s_1=1$. Let $f\in \Sol(I)$ be a solution to $I$ and let $F=(f,s_2\bullet f, \ldots, s_m \bullet f)$. Since $\rank(I)=m$, there exist unique matrices $P_1,\ldots,P_n\in \CC(x_1,\ldots,x_n)^{m\times m}$ such that 
\begin{align}\label{eq:Pfaffsys}
\partial_i \bullet F \,=\, P_i \cdot F \qquad \text{for }\ i=1,\ldots,n \, .
\end{align}
The matrices $P_i$ can be computed via a Gr\"{o}bner basis reduction (cf.~\cite[p.~23]{SatStu19}).
The system in \eqref{eq:Pfaffsys} is the {\em Pfaffian system of~$I$} for the chosen term ordering on the Weyl algebra. The matrices $P_i$ obey the integrability condition, which translates as 
\begin{align}\label{eq:intmatrices}
P_i \cdot P_j - P_j \cdot P_i \ =\ \partial_i \bullet P_j - \partial_j \bullet P_i \qquad \text{for all } i,j \, ,
\end{align}
On the right hand side of \eqref{eq:intmatrices}, entry-wise differentiation the matrices $P_i$ is meant.

\subsection{Indicial ideals and the Nilsson ring}
We first recall some results from \cite[Sections 2.2--2.6]{SST00} in Theorems~\ref{thm:rankinw}--\ref{thm:liftalgo} below, in order to present the overall strategy to compute solutions of our $D$-ideals. 
Here and throughout the rest of the file, weight vectors for the Weyl algebra are allowed to be taken from 
\begin{align}\label{eq:weightsD}
W \,\coloneqq \, \left\{ (u,v)\in \RR^{2n} \, |\, u_i+v_i \geq 0, \, i=1,\ldots,n \right\}\,.
\end{align}
Among others, $W$ contains the set $\{(-w,w) | w\in \RR^n\}$.
For a weight vector $(u,v)\in \RR^{2n}$ and $P\in D_n$, the {\em initial form} $\init_{(u,v)}(P)$ of $P$ denotes the part of $P$ of {\em maximal} $(u,v)$-weight. For weights of the form $(-w,w)$, one denotes the initial form simply by $\init_w(P)$. Each weight vector $(u,v)$ induces an increasing filtration $F_{(u,v)}^\bullet (D_n)$ of the Weyl algebra by the $(u,v)$-weight. For weights of the form $(-w,w)$, the associated graded ring \begin{align}
\operatorname{gr}_{(-w,w)}\left(D_n\right) \,=\, \bigoplus_k \left( F_{(-w,w)}^k\left(D_n\right)/F_{(-w,w)}^{k-1}\left(D_n\right)\right) 
\end{align}
is isomorphic to the Weyl algebra itself, cf.~\cite[p.~4]{SST00}. For $(u,v)=(0,1)\in \RR^{2n}$, the associated graded ring $\operatorname{gr}_{(0,1)}(D_n)$ is the polynomial ring $\CC[x_1,\ldots,x_n][\partial_1,\ldots,\partial_n]$, in which case one typically replaces $\partial_i$ by $\xi_i$ to highlight 
the commutativity of the variables. The initial ideal $\init_{(u,v)}(I)$ of a $D_n$-ideal $I$ is the ideal generated by $\init_{(u,v)}(P)$ of all $P\in I$. 
It naturally lives in 
the graded ring \mbox{$\operatorname{gr}_{(u,v)}(D_n)$}. 
The initial ideal for the weight $(0,1)\in \RR^{2n}$ is exactly the characteristic ideal of the $D_n$-ideal, which was introduced in \Cref{sec:Didealsols}.

\medskip 
In what follows, we focus on weights of the form $(-w,w)$. The following theorem summarizes Theorems~2.2.1 and~2.5.1 of \cite{SST00} about the holonomic rank of the~initial~ideal.

\begin{theorem}\label{thm:rankinw}
Let $I$ be a holonomic $D_n$-ideal and $w$ any weight vector in $\RR^n$. Then also the initial ideal $\init_{(-w,w)}(I)$ is holonomic and $\rank(\init_{(-w,w)}(I))\leq \rank(I).$
If $I$ is regular holonomic, then $\rank(I)=\rank(\init_{(-w,w)}(I)).$
\end{theorem}

\begin{theorem}[{\cite[Theorem 2.3.3]{SST00}}]\label{thm:torusfixed}
Let $I$ be a $D_n$-ideal.
\begin{enumerate}[(1)]
\item $I$ is torus-fixed if and only if $\init_{(-w,w)}(I)=I$ for all $w\in \RR^n$.
\item If $w$ is generic for $I$, then $\init_{(-w,w)}(I)$ is torus-fixed.
\end{enumerate}
\end{theorem}

\noindent For the precise definition of what ``generic'' means here, we refer to the later \Cref{sec:Grfan}.

We are going to exploit that distractions of torus-fixed ideals (such as indicial ideals) take on a very special form.
The {\em indicial ideal} of $I$ with respect to $w$ is the distraction of the initial ideal, i.e., the $\CC[\theta_1,\ldots,\theta_n]$-ideal
\begin{align}\label{eq:indicial}
\ind_w(I) \, \coloneqq \,  R_n \cdot \init_{(-w,w)}(I)\cap \CC[\theta_1,\ldots,\theta_n],
\end{align}
where $\theta_i=x_i\partial_i$ denotes the $i$th Euler operator. 

\begin{definition}\label{def:exponents}
The zeros 
\begin{align}
    V\left(\ind_w(I)\right) \,=\, \left\{ A \in \CC^n \,|\, q(A)=0 \text{ for all } q\in \ind_w(I) \right\} 
\end{align}
of the indicial ideal are called the {\em exponents of $I$ with respect to~$w$}. 
\end{definition}

A $D_n$-ideal $F$ which is generated by elements of $\CC[\theta_1,\ldots,\theta_n]$ is called {\em Frobenius ideal}. Frobenius ideals hence are of the form $F=D_nJ$ with $J$ an ideal in $\CC[\theta_1,\ldots,\theta_n]$. A Frobenius ideal $F=D_n J$ is holonomic if and only if $J$ is Artinian, i.e., if $\CC[\theta]/J$ is a finite-dimensional $\CC$-vector space. Now assume $J$ is Artinian. The primary decomposition of $J$ is of the form
\begin{align}\label{eq:primdecFrob}
J \,=\, \bigcap_{A\, \in\, V(J)} Q_A(\theta-A) \, ,
\end{align}
where $Q_A$ is an Artinian ideal, and $Q_A(\theta-A)$ denotes the ideal obtained from replacing each $\theta_i$ by $\theta_i-A_i$ in $Q_a$.
Solutions of Frobenius ideals take on the very special form $x^A\cdot g(\log(x))$, which can be read from the primary 
decomposition of $J$. Namely, $A\in V(J)$ and $g$ runs over the finite-dimensional $\CC$-vector space 
\begin{align}\label{eq:orthcompl}
\begin{split}
Q_A^{\perp} \,=\, \left\{ p \in \CC[x_1,\ldots,x_n] \,| \, f(\partial_1,\ldots,\partial_n) \bullet p(x_1,\ldots,x_n)=0 \right. \\
\left. \text{ for all } f=f(\theta_1,\ldots,\theta_n )\in Q_A\right\} ,
\end{split}
\end{align}
the {\em orthogonal complement}
of the Artinian ideal $Q_A$.
Since we will make use of that strategy regularly later on, we summarize what was said above in the following proposition.

\begin{proposition}\label{prop:solsind}
Let $F=D_n J,$ where $J\subset \CC[\theta]$, be a holonomic Frobenius ideal.  
The solution space of $F$ is spanned by the functions $x^A\cdot g(\log(x))$, where $A$ runs over the points of the variety~$V(J)$, and $g$ runs over the orthogonal complement~$Q_A^\perp$ from~\eqref{eq:orthcompl}.
\end{proposition}

\begin{proposition}
Let $I$ be a holonomic ideal and $w\in \RR^n$generic for $I$. Then
$\ind_w$(I) is a holonomic Frobenius ideal whose rank equals $\rank(\init_{(-w,w)} (I))$.
\end{proposition}

\medskip

By $N$, we denote the ring of functions of the {\em Nilsson class}, i.e., those functions which can be represented by an element of \begin{align}\label{eq:defNilsson} 
N \,\coloneqq \, \CC\llbracket x^{u^1},\ldots,x^{u^n}\rrbracket[x^{\beta^1},\ldots,x^{\beta^n},\log(x_1),\ldots,\log(x_n)]
\end{align}
for suitable vectors $u^1,\ldots,u^n,\beta^1,\ldots,\beta^n\in \CC^n$ (see \cite[(2.31)]{SST00}). The coefficients lie in the ring $\CC\llbracket x^{u^1},\ldots,x^{u^n} \rrbracket$ of formal power series in the $x^{u^i}$'s.

\begin{definition}\label{def:wweightmon}
The {\em $w$-weight} of a monomial $x^A\log(x)^B$ is the real part $\Re(w\cdot A)$ of~$w\cdot A$.
\end{definition}
The {\em initial series} of a function $f=\sum_{A,B} c_{AB}x^A\log(x)^B\in N$, denoted $\init_w(f)$, is defined to be
 the finite subsum of all terms of {\em minimal} \mbox{$w$-weight}.
Note that a weight vector $w\in \RR^n$ induces a partial order on functions of the Nilsson class as follows:
\begin{align}
x^A\log(x)^B  \leq x^c\log(x)^d \ \Leftrightarrow \ \Re(w \cdot A ) \leq \Re(w \cdot c) \, .
\end{align}
Since the $w$-weight does not give a monomial ordering, one needs a monomial order $\prec$ as a tie breaker and denotes the resulting monomial order by~$\prec_w$. We will take $\prec$ to be the lexicographical ordering on $N$ obtained as restriction of the lexicographic ordering~on~$\CC^n \oplus \NN^n$.
 
The set of {\em starting monomials} of $I$ with respect to $\prec_{w}$ is
\begin{align}
\Start_{\prec w}(I) \, \coloneqq \, \left\{ \init_{\prec w}(f) \,|\, f\in N \text{ is a non-zero solution of } I\right\},
\end{align}
where $\init_{\prec w}(f)=x^A\log(x)^B$ for some $A\in \CC^n$ and $B\in \NN^n$. Here, $x^A$ denotes $x_1^{A_1}\cdots x_n^{A_n}$ and $\log(x)^B=\log(x_1)^{B_1}\cdots \log(x_n)^{B_n}$. Moreover,
\begin{align}
\Start_{\prec_w}(I) \,=\, \Start_{\prec_w}(\init_{(-w,w)}(I))\,=\, \Start_{\prec_w}(\ind_{w}(I)) \, .
\end{align}

\begin{proposition}[{\cite[Corollary 2.5.11]{SST00}}]
If $x^A \log(x)^B \in  \Start_{\prec w}(I)$, then $A$ is an exponent of $I$ with respect to~$w$. For each exponent $A$, the number of starting monomials of the form $x^A \log(x)^B$ is the multiplicity of $A$ as a root of the indicial ideal $\ind_w(I)$.
\end{proposition}

\begin{theorem}[{\cite[Theorem 2.5.12]{SST00}}]\label{thm:lift}
For each starting monomial $x^A\log(x)^B$ in $ \Start_{\prec w}(I) = \Start_{\prec_w}(\ind_{w}(I))$, there exists a unique $f\in N$ with the following properties: 
\begin{enumerate}[(1)]
\item $f$ is annihilated by $I$, i.e., $f\in \Sol(I)$.
\item $\init_w(f) = x^A\log(x)^B$.
\item The monomial $x^A\log(x)^B$ is the only starting monomial that appears in $f$ with non-zero coefficient.
\end{enumerate}
\end{theorem}
Solutions to $I$ as in the theorem above are called {\em canonical (series) solutions} of $I$ with respect to~$\prec_w$. The solution functions $f$ in \Cref{thm:lift} can be shown to actually live in the Nilsson ring $N_w\subset N$, which is the content of the next proposition.

\begin{proposition}
If $I$ is regular holonomic and $w$ a generic weight for $I$, there exist $\rank(I)$ many canonical series solutions of $I$ which lie in the Nilsson ring $N_w(I)$ of $I$ w.r.t.\ to $w$,
\begin{align}\label{eq:Nilssonw}
N_w(I) \, \coloneqq \, \CC \llbracket {C_w(I)}_{\ZZ}^*\rrbracket [x^{e^1},\ldots,x^{e^r}, \log (x_1),\ldots ,\log(x_n)] \, .
\end{align}
\end{proposition}
\noindent Here, $\{e^1,\ldots,e^r\}$ denotes the set of roots of the indicial ideal of $I$, and $C_w(I)_{\ZZ}^*=C_w(I)^*\cap \ZZ^n$, where $C_w(I)$ is the {\em Gr\"{o}bner cone} of $I$ containing $w$,
\begin{align}\label{eq:Grcone} 
C_w(I) \,=\, \left\{ w' \in \RR^n | \init_{(-w,w)}(I) = \init_{(-w',w')}(I) \right\},
\end{align}
and ${C_w(I)}^*=\{ u \in \RR^n \mid u \cdot v \geq 0\text{ for all } v \in C_w \}$ is its dual cone (called ``polar dual'' in~\cite{SST00}). The elements of $\CC \llbracket {C_w(I)}_{\ZZ}^*\rrbracket$ are power series in $x$ whose exponent vectors lie in $C_w(I)_{\mathbb{Z}}^{\ast}$.

More precisely, the canonical solutions to $I$ with respect to $\prec_w$ have the form $x^A\cdot g$, where $A$ is an exponent of $I$ and $g$ is an element of $\CC \llbracket {C_w(I)}_{\ZZ}^*\rrbracket [ \log (x_1),\ldots ,\log(x_n)] ,$ such that the degree of each $\log(x_i)$ in $g$ is at most $\rank(I)-1$ (see \cite[Theorem 2.5.14]{SST00}).

\begin{definition}\label{def:Grweight}
Let $w\in \RR^n$. A finite subset $G$ of a $D$-ideal $I$ is called a {\em Gr\"{o}bner basis with respect to $w$} if $I$ is generated by $G$ and $\init_{(-w,w)}(I)\subset D$ is generated by the initial forms $\init_{(-w,w)}(g)$, where $g$ runs over the set $G$.
\end{definition}

Moreover, if $G$ is a Gr\"{o}bner basis of $I$ with respect to $\prec_w$, where $\prec$ denotes any term order on $D$ (see \cite[p.5]{SST00}), $G$ is also a Gr\"{o}bner basis of $I$ with respect to $w$. 

\subsection{The Saito--Sturmfels--Takayama algorithm}
We now state the theorem which allows us to compute canonical solutions starting from a Gr\"{o}bner basis of $I$ with respect to a generic weight vector $w\in \RR^n$. This is {\cite[Theorem~2.6.1]{SST00}}, and we here recall it together with its proof, which contains the algorithm to lift solutions of the indicial ideal to canonical series solutions of~$I$ up to arbitrary $w$-weight. We are going to refer to this algorithm as {\em SST algorithm}.

\begin{theorem}\label{thm:liftalgo}
Let $I$ be a regular holonomic ideal in $\mathbb{Q}[x_1,\ldots,x_n]\langle \partial_1,\ldots,\partial_n \rangle$ and let \mbox{$w\in \RR^n$} be a generic weight vector for $I$. Let $I$ be given by a Gr\"{o}bner basis $G$ w.r.t.~$w$. There exists an algorithm which computes all terms up to specified $w$-weight in the canonical series solutions to $I$ with respect to $\prec_w$.
\end{theorem}

For the convenience of our readers, we summarize the main steps of the SST algorithm as a procedure before turning to the proof. We also give an example of the algorithm running on a one-variable hypergeometric system. We here already mention Gröbner fans; the definition is contained in the later \Cref{sec:Grfan}.

\begin{procedure}[Computing canonical series solutions of a $D$-ideal up to a chosen order]\label{proc:SST}

\noindent {\em Input:} A regular holonomic $D_n$-ideal $I$, its small Gr\"{o}bner fan $\Sigma$ in $\RR^n$, a weight vector $w\in \RR^n$ that is generic for $I$, and the desired order $k\in \NN$. 

\begin{enumerate}[{\em Step 1.}]
\item Compute the indicial ideal $\ind_w(I)$ and its $\rank(I)$ many solutions. They are the form $x^A\log(x)^B$ with $A \in V(\ind_w(I))$, and will be the starting monomials of the canonical series solutions. For each starting monomial, carry out Steps 2--5. 
\item Determine a Gr\"{o}bner basis $G$ of $I$ with respect to the weight~$w$. 
\item Create the recurrences. By a recurrence, we mean a way of writing each element $g\in G$ as $x^\alpha g = f-h$ with $\alpha \in \mathbb{Z}^n$ an integer vector such that $f\in \KK[\theta_1,\ldots,\theta_n]$ and $h \in \KK[x_1^{\pm 1},\ldots,x_n^{\pm 1}]\langle \partial_1,\ldots,\partial_n \rangle $ with $\operatorname{ord}_{(-w,w)}(h)<0 $.\footnote{It will become clear in \Cref{eg:hypergeo} why we call such a representation a recurrence.} For each recurrence, write $h = \sum_{j=1}^r x^{\beta(j)}h_j$ where each $h_j$ is in $\KK[\theta_1,\ldots,\theta_n].$ Write $f, h_1, ..., h_r$ as matrices acting on the vector spaces $L_p,L_{p - \beta(1)}, \ldots, L_{p - \beta(r)},$ which contain the coefficient vectors $c_p$ of $x_p\log(x)^b,$ $ 0 \leq b_i < \rank I.$ 
\item Apply the recurrences. Assuming one has the coefficient vectors $c_{p - \beta(1)}, ..., c_{p - \beta(r)},$ this amounts to solving the system of matrix equations to obtain $c_p.$ 
\item If the matrix of $f$ is singular, then one may need to use the condition that for $i \neq j,$ the series expansion $s_i$ must have coefficient $0$ for the starting monomial of $s_j.$ 
\end{enumerate}
\noindent {\em Output:} The canonical series solutions of $I$ with respect to $w$, truncated at $w$-weight $k$.
\end{procedure}

\begin{example}\label{eg:hypergeo}
We apply this procedure to the ideal $I \subset \CC[x]\langle \dx \rangle$ generated by the hypergeometric operator $P  \coloneqq  \theta(\theta- 3) - x(\theta+ a)(\theta+ b).$  Its Gr\"{o}bner fan has only two maximal cones, namely the rays $\pm \mathbb{R}_{\geq 0}$. We choose the weight vector $w = 1.$ The ideal $I$ is univariate and principal, hence holonomic of rank $\ord_{(0,1)}(P)  = 2.$

\begin{enumerate}[{\em Step 1.}]
\item The order of $x$ with respect to $w = 1$ is $-1.$ Thus $\init_{(-w,w)}(I)=\langle \theta(\theta- 3)\rangle.$ The initial ideal is already torus-fixed, so $\ind_w(I) = \init_{(-w,w)}(I).$ Solutions to the indicial ideal are $x^0 =1$ and $x^3.$ We select the starting monomial $x^3.$ Our coefficient vectors~$c_p$ are in $L_p  = \CC \cdot \{x^{p+3}, x^{p+3} \log(x)\}$. 
\item The ideal is principal, hence its generator is a Gr\"{o}bner basis of the ideal.
\item  We write $P = f-h,$ where $f = \theta(\theta- 3)$ and $h = x(\theta+ a)(\theta+ b).$ It suffices to compute the action of $\theta$ on each element of $L_p$ and extend  it $k[\theta]$-linearly. We have
\begin{align*}
\theta \bullet x^{p+3}  \,=\, (p+3)x^{p+3} \quad \text{and}\quad \theta \bullet x^{p + 3}\log(x)  \,=\, x^{p + 3} + (p+3)x^{p + 3}\log(x) \,.
\end{align*}
Thus, the matrix of the operator $\theta$ in the basis $\{ x^{p+3}, x^{p+3} \log(x)\}$ is 
\begin{align*}
    \begin{bmatrix}
    p + 3 & 1\\
    0 & p + 3
    \end{bmatrix}.
\end{align*} 

\item Let $c_{p,1}$ and $ c_{p,2}$ be the coefficients of $x^{p+3}$ and $x^{p+3}\log(x)$ in the power series expansion. Then we can write our operators as matrices, and our recurrence as
\begin{align*}
    \begin{bmatrix}
    p & 1\\
    0 & p
    \end{bmatrix}
    \begin{bmatrix}
    p + 3 & 1\\
    0 & p + 3
    \end{bmatrix}
    \begin{bmatrix}
    c_{p,1}\\
    c_{p,2}
    \end{bmatrix}
    \,=\,
    \begin{bmatrix}
    p - a + 2 & 1 \\
    0 & p - a + 2
    \end{bmatrix}
    \begin{bmatrix}
    p - b + 2 & 1 \\
    0 & p - b + 2
    \end{bmatrix}
    \begin{bmatrix}
    c_{p-1,1}\\
    c_{p-1,2}
    \end{bmatrix}
\end{align*}
with initial values $c_0 = 1, \, d_0 = 0.$ Solving the recurrence yields the explicit~formulae
\begin{align*}
    c_{p,1} = 0 \,  \quad \text{and} \quad c_{p,2} = \frac{(a+3)_p(b+3)_p}{(1)_p(4)_p} \, ,
\end{align*} 
where $(a)_p = a(a + 1)\cdots(a + p -1)$ is the Pochhammer symbol.
\item If we choose the starting monomial $x^0=1$ instead, the matrix of $f$ is singular for $p = 3.$ We leave it to the reader to find the series expansion, or to see \cite[pp. 98--99]{SST00}.
\end{enumerate}
\end{example}

We now turn to the proof of \Cref{thm:liftalgo}, making the steps of \Cref{proc:SST} precise.

\smallskip

{\noindent {\em Proof and algorithm.}} Let $w\in \RR^n$ be generic for $I$. Compute the roots of $\ind_w(I)$ and extend the field of coefficients by them. Denote the resulting, computable field extension of $\QQ$ by $\mathbb{K}$. To compute the canonical solution of $I$ whose starting monomial is $x^A\log(x)^B$, one proceeds as follows. For $p\in \ZZ^n$, denote by $L_p$ the $\KK$-vector space
\begin{align}\label{def:Lp}
L_p \, \coloneqq \, x^A  \cdot \hspace*{-5mm} \sum_{0\leq b_i < \rank(I)} \KK \cdot x^p \, \log(x)^b \, .
\end{align} 
The $\KK$-vector space $L_p$ is finite-dimensional. Every $f\in \KK[\theta]$ induces a $\KK$-linear map \mbox{$f\colon L_p\to L_p$}. The monomials of $L_p$ are a $\KK$-basis of it. They are ordered by the term order~$\prec_w$ on the Nilsson ring, starting with the smallest. The matrix of $f$ in this basis is an upper triangular square matrix. Let $L_p'$ denote the set of monomials in $L_p$ that are not contained in $\Start_{\prec_w}(I)$. 
Now let $\{f_1,\ldots,f_d\}$ be any generating set of $\ind_w(I)$ and restrict $f_i\colon L_p\to L_p$ to $L_p'$; this corresponds to deleting some of the columns in the associated matrix. Denote the resulting matrix by $F_i$. Then the map
\begin{align}\label{eq:mapfi}
F: L_p'\longrightarrow L_p^d \, , \quad v\mapsto (f_1\bullet v,\ldots,f_d \bullet v)^{\top}
\end{align}
is injective and is represented by the matrix obtained as vertical concatenation of $F_1,\ldots,F_d$.

Now let $G = \{g_1, ..., g_d\}$ be a Gr\"{o}bner basis of $I$ with respect to $w$; its Gr\"{o}bner cone in~$\RR^n$ is denoted by $C_w$. 
For each $g \in G$, choose a Laurent monomial $x^{\alpha}$ such that
\begin{align}\label{eq:gfh}
x^{\alpha}g \,=\, f - h \, ,
\end{align}
where $f\in \KK[\theta_1,\ldots,\theta_n]$,  $h\in \KK[x_1^{\pm 1},\ldots,x_n^{\pm 1}]\langle \partial_1,\ldots,\partial_n \rangle $ with $\operatorname{ord}_{(-w,w)}(h)<0 $. Here, $\operatorname{ord}_{(-w,w)}(h)$ denotes the largest $w$-weight of a monomial appearing in $h$. Then the set of operators $f_1, \ldots , f_d$ obtained this way generate $\ind_w(I)$ and, as maps as in~\eqref{eq:mapfi}, they are injective. 
The Laurent monomial $x^\alpha$ for a Gr\"{o}bner basis element $g$ is obtained by taking a highest-weight term $x^{a} \partial^{b} m$ in $g,$ where $m$ is a monomial in  the $\theta_i$'s and for all $k$, at least one of $\{a_k, b_k\}$ is zero. Intuitively, this corresponds to pulling out as many $\theta$'s as possible into~$m$. Then $x^\alpha =  x^{b-a}$. The $h_i$ may have terms of different weights; in this case, we get a recurrence that involves $L_p$ for more than two different $p$'s. 
The coefficients of the canonical series solution are now computed by induction on the $w$-weight $k$. We start from a canonical series solution $x^A\log(x)^B+\cdots $ and assume the coefficients $c_{pb}$ are already known for $0\leq p\cdot w\leq k,$  $p\in C^*_{\mathbb{Z}}$.
Let $F_k(x)$ be this partial solution up to $w$-weight~$k$, i.e., 
\begin{align}\label{eq:Fk}
F_k(x) \,=\ x^A \cdot \hspace*{-5mm} \sum _{\substack{0\leq p\cdot w\leq k,\, p\in C^*_{\mathbb{Z}}, \\ 0 \leq b_j < \rank(I)}} c_{pb} \, x^p\log(x)^b \,,
\end{align}
and let $M_k$ be the space of terms with $w$-weight greater than $k$, i.e.\ $M_k = \sum_{p \cdot w > k, p \in C^*_{\ZZ}} L_p$. Then, by definition, we have that
\begin{align} \label{eq:goFk}
g_i \bullet F_k(x) \,\equiv\, 0 \mod M_k \, .
\end{align} 
Assuming we know $F_k(x)$ for some $k$, we are going to construct a recursion which allows us to determine the additional terms which are needed to lift $F_k(x)$ to $F_{k+1}(x)$. The starting point of that recursion will be the starting monomials. 
We hence look for an element $E_{k+1}$ of $\sum_{p \cdot w = k + 1, p \in C^*_{\ZZ}}L'_p$ such that
\begin{align}
 E_{k+1}(x) \, \equiv \, F_{k+1}(x) - F_{k}(x) \mod M_{k+1} \,.
\end{align}
To achieve this, observe that $\operatorname{ord}_{(-w,w)}(h_i)<0$ implies that $h_i \bullet x^{\ell}$ has \emph{higher} $w$-weight than~$x^{\ell}$. We can show this on monomials as follows. Suppose that $h_i = x^q$ for some $q\in \NN^n$ with $q \cdot (-w) < 0.$ Then $q \cdot w >0,$ so $h_i\bullet x^{\ell}=x^{\ell + q}$ has higher $w$-weight than $x^{\ell}.$ Similarly, suppose that $h_i = \partial^r$ where $r \cdot w < 0.$ Then $(-r) \cdot w > 0.$ Thus $h_i \bullet x^{\ell} = Cx^{\ell-r}$ for some constant $C,$ and $x^{\ell-r}$ has higher $w$-weight than $x^{\ell}$. Together with~\eqref{eq:gfh} and~\eqref{eq:goFk}, this implies that $f_i \bullet F_{k + 1}$ = $h_i \bullet F_k \mod M_{k+1}$, which gives the desired recursion relation for $E_{k+1}(x)$ in terms of $F_k(x)$, namely
\begin{equation}\label{eq:remainder}
f_i \bullet E_{k+1}(x) \, \equiv \, (h_i - f_i) \bullet F_k(x) \mod M_{k+1} \,.
\end{equation}
By the injectivity of the map $F$ from \eqref{eq:mapfi} and the existence of a canonical series solution, there exists a unique solution $E_{k+1}(x)$ to \eqref{eq:remainder}, and this lifts $F_k(x)$ to $F_{k+1}(x)$. \hfill \qed

\begin{remark}\label{rem:choicealpha}
The choice of $\alpha$ in \eqref{eq:gfh} is in general not unique, and sometimes there might not be an $\alpha$ at all which satisfies all of the conditions. For example, consider $ w =  (-1, 0)$ and the operator $P=\dx_2 + x_2\dx_2 + x_2^2\dx_2 + \dx_1 + x_2\dx_1$. All its monomials have weight $0$~or~$-1.$ Also, no matter what Laurent monomial we multiply by, only one of $x^\alpha \dx_2, x^\alpha x_2\dx_2,$ and $x^\alpha x_2^2\dx_2$ can be torus-fixed. So $h$ will contain a term which has weight $0,$ and the order of $h$ will always be the same as the order of $f$. Thus no choice of $\alpha$ as required exists. For a case where several choices for the monomial $x^\alpha$ exist, take e.g.\ $\theta_1 + \theta_2 + \theta_3 + 1.$ Then any~$x^\alpha$ with negative $w$-weight will work. 
However, we note that two different monomials having the same weight only occurs for a measure zero subset of the Gr\"obner cone. Thus, if we have picked a generic weight vector $w,$ we can perturb it slightly to $w'$ for which a choice of $\alpha$ exists, while remaining in the same Gr\"obner cone. In the cases we consider later on, there is a unique choice for $\alpha.$
\end{remark}

\section{Computing solutions of \texorpdfstring{$I_3$}{I\_3} with Gr\"{o}bner techniques}\label{sec:solI3Groebner}
In this section, we compute the canonical series solutions of the $D$-ideal $I_3$ generated by the operators in \eqref{eq:p1p2p3tilde} using the SST algorithm from the proof of \Cref{thm:liftalgo}. We begin in \Cref{sec:Grfan} by showing that $I_3$ is holonomic, and by computing some preliminary data, such as its holonomic rank, its singular locus, and its Gr\"obner fan. In \Cref{sec:Groebnerw1}, we discuss thoroughly the computation of the canonical series solutions for one of the cones of the Gr\"obner fan. Thanks to the symmetry of the system, the computations for the remaining two cones work analogously; we give more details for that at the end of the section.
 
\subsection{The Gr\"{o}bner fan of \texorpdfstring{$I_3$}{I\_3}}\label{sec:Grfan}
In order to employ the SST algorithm for computing the canonical series solutions, we first need to certify that the ideal is holonomic. Let $I_3(c)$ be the $D$-ideal generated by the operators in \eqref{I3system}, with $c=(c_0,c_1,c_2,c_3)$. A computation in {\sc Singular}~\cite{Singular,Plural} approves that $I_3(4,2,2,2)$ is holonomic.
The output of the {\sc Singular} code below proves that $I(c)$ is holonomic also for generic~$c$. For that, we use the Weyl algebra $\QQ (c_0, c_1, c_2, c_3) [x_1, x_2, x_3]\langle \partial_1, \partial_2, \partial_3\rangle$ with the degree reverse lexicographical order for\linebreak $x_1 > x_2 > x_3>\partial_1 > \partial_2 >  \partial_3.$ 
The code uses the $D$-module libraries~\cite{ABLMS}, and can be carried out by running the following lines, which also compute the holonomic rank, the singular locus, and the characteristic variety~of~$I_3(c)$. 

{\small
\begin{verbatim}
LIB "dmod.lib";
ring r = (0,c0,c1,c2,c3),(x1,x2,x3,d1,d2,d3),dp; setring r;
def D3 = Weyl(r); setring D3;
poly P1 = 4*x1*d1^2 - 4*x3*d3^2 + (2*c0-4*c1+4)*d1 + (-2*c0+4*c3-4)*d3;
poly P2 = 4*x2*d2^2 - 4*x3*d3^2 + (2*c0-4*c2+4)*d2 + (-2*c0+4*c3-4)*d3;
poly P3 = 2*x1*d1 + 2*x2*d2 + 2*x3*d3 + (2*c0-c1-c2-c3);
ideal I = P1,P2,P3;
isHolonomic(I); holonomicRank(I); DsingularLocus(I);
def CV = charVariety(I); setring CV; charVar;
\end{verbatim}
}
\noindent
Since we compute over the field of rational functions in the parameters $c,$ the  results are valid for generic~$c$.
The characteristic ideal $\init_{(0,1)}\left(I_3(c)\right)$ is the  $\CC[x_1,x_2,x_3][\xi_1,\xi_2,\xi_3]$-ideal generated by the $8$ operators
\begin{align}
\begin{split}
x_1\xi_1 + x_2\xi_2 + x_3\xi_3, \ x_2\xi_2^2- x_3\xi_3^2,\ x_2\xi_1\xi_2 + x_3\xi_1\xi_3 + x_3\xi_3^2, \\ 2x_2 x_3  \xi_2  \xi_3-x_1 x_3  \xi_3^2+x_2 x_3  \xi_3^2+x_3^2  \xi_3^2, \ 
x_3  \xi_1  \xi_2  \xi_3+x_3  \xi_1  \xi_3^2+x_3  \xi_2  \xi_3^2, \\
2x_1 x_3  \xi_2  \xi_3^2-2x_3^2  \xi_2  \xi_3^2-x_1 x_3  \xi_3^3+x_2 x_3  \xi_3^3-3x_3^2  \xi_3^3, \\
2x_2 x_3  \xi_1  \xi_3^2-2x_3^2  \xi_1  \xi_3^2+x_1 x_3  \xi_3^3-x_2 x_3  \xi_3^3-3x_3^2  \xi_3^3, \\
x_1^2 x_3  \xi_3^3-2x_1 x_2 x_3  \xi_3^3+x_2^2 x_3  \xi_3^3-2x_1 x_3^2  \xi_3^3-2x_2 x_3^2  \xi_3^3+x_3^3  \xi_3^3 \, ,
\end{split}
\end{align}
all of which are independent of $c$. The singular locus is
\begin{equation}
\operatorname{Sing}\left(I_3(c)\right) \,=\, V\left( x_1x_2x_3\cdot  \lambda \right) \, \subset \, \CC^3,
\end{equation}
where $\lambda$ is the homogeneous degree-$2$ polynomial defined in \Cref{eq:lambda}. 
Note that the singular locus is symmetric in each of the~$x_i$'s, and is in agreement with what was observed for the known solutions~\eqref{eq:SingIfs}.

Our computations in {\sc Singular} show that $I_3(c)$ has holonomic rank $4.$ Therefore, by the theorem of Cauchy, Kovalevskaya, and Kashiwara, the $\CC$-vector space of holomorphic solutions to $I_3(c)$ on a simply connected domain in $\CC^3\setminus\Sing\left(I_3(c)\right)$ is $4$-dimensional. We now focus on $I_3 \coloneqq I_3(4,2,2,2)$ and compute its Gr\"obner fan. Again, denote by~$W$ the set of weight vectors for~$D_n$ from \Cref{eq:weightsD}, i.e., $W =\left\{ (u,v) \,|\, u+v \geq 0 \right\}  \subset  \RR^{2n}.$
The {\em Gr\"{o}bner fan} of a $D_n$-ideal $I$ is a finite polyhedral fan in $\RR^{2n}$ with support $W$ such that the initial ideal $\init_{(u,v)}(I) \subset \operatorname{gr}_{(u,v)}(D_n)$ is constant as $(u,v)$ ranges over any of its cones. The {\em small Gr\"{o}bner fan} is the restriction of the Gr\"{o}bner fan to $\{ (u,v) \,|\, u+v=0\} \simeq \RR^n.$ Hence, the cones of the small Gr\"{o}bner fan of $I$ are in one-to-one correspondence with the initial ideals $\init_{(-w,w)}(I)$, with $w\in \RR^n$.
\begin{definition}\label{def:genericw}
A weight vector  is {\em generic for $I$} if it lies in an open cone of the small Gr\"{o}bner fan of $I$.
\end{definition}
To compute the small Gr\"{o}bner fan of our system $I_3$, our strategy is to first approximate the fan by computing the initial ideals for various lattice vectors, and finding where they change. Once we have a good enough approximation to guess our fan, we can verify our guess by computing initial ideals along the candidate rays. The system $I_3$ is homogeneous with respect to the vector $(1, 1, 1).$ Thus, our small Gr\"{o}bner fan really lives in $\RR^3 / (1, 1, 1).$ More specifically: if our Gr\"{o}bner fan is~$\Sigma$, then $\Sigma  = \Sigma' \times (1, 1, 1)$ where $\Sigma'$ is a fan in $\RR^2.$ We choose the orthogonal basis $v_1 = (1, 0, -1)$ and $v_2 = (-1, 2, -1)$ of $\RR^3/ (1,1,1)$. By explicitly computing the initial ideals of $I_3$, we obtain $\Sigma'$ as depicted in \Cref{fig:grobnerfan}. Note that this picture is not canonical; it depends on the choice of basis of $\RR^3/(1,1,1).$

\medskip

The $D$-ideal $I_3$ is moreover regular singular, as we will read later from the system~\eqref{eq:By}.

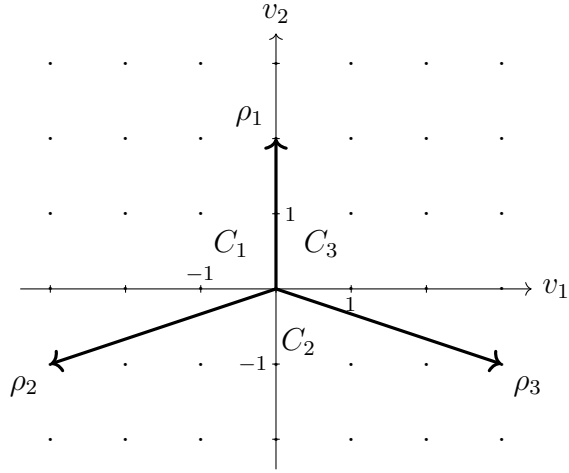
\begin{figure}
\begin{tikzpicture}[scale=1.0, transform shape]
\tikzstyle{grid lines}=[black,line width=0]
\foreach \r in {-3,...,2} \draw[black] (\r,-0.05) -- (\r,0.05);
\foreach \r in {-2,...,3} \draw[black] (-0.05,\r) -- (0.05,\r);
\foreach \r in {-3,...,3}
\foreach \s in {-2,...,3}
\node[fill=none,draw=none] at (\r,\s) {.};
\draw[black,->] (0,-2.4)  -- (0,3.4) node[above] {$v_2$};
\draw[black,->]  (-3.4,0)  -- (3.4,0) node[right] {$v_1$};
\draw[black,->,very thick] (0,0)  -- (0,2) node[above left] {$\rho_1$};
\draw[black,->,very thick] (0,0)  -- (-3,-1) node[below left ] {$\rho_2$};
\draw[black,->,very thick] (0,0)  -- (3,-1) node[below right] {$\rho_3$};
\node at (-0.6,0.6) {$C_1$};
\node at (0.3,-0.7) {$C_2$};
\node at (0.6,0.6) {$C_3$};
\node at (1.0,-0.2) {${}_1$};
\node at (0.2,1) {${}_1$};
\node at (-0.3,-1) {${}_{-1}$};
\node at (-1.0,0.2) {${}_{-1}$};
\end{tikzpicture}
\caption{The fan $\Sigma'$ in the basis $(v_1,v_2)$, where $v_1 = (1,0,-1)$ and $v_2 = (-1,2,-1)$. 
Its rays are $\rho_1=\RR_{\geq 0}\cdot (0,1)$, $\rho_2=\RR_{\geq 0}\cdot(-3, -1)$, and $\rho_3=\RR_{\geq 0}\cdot(3,-1)$. 
In $\RR^3$, the corresponding hyperplanes are $\rho_1 = [\{w_1 - w_3 = 0\}],$ $ \rho_2 =[ \{w_1 - w_2 = 0\}],$ and $\rho_3 = [\{w_2 - w_3 = 0\}].$ The cones $C_1$, $C_2$, and $C_3$ correspond to the following cones in $\RR^3$: $ C_1 =[\{w_1<w_2,w_3\}]$, $C_2=[\{ w_2<w_1,w_3\}]$, and $C_3 =[\{w_3<w_1,w_2\}]$.}
\label{fig:grobnerfan}
\end{figure}

\begin{remark}
The triangle integral---read as a function of the coefficients of the graph polynomial---is a solution of a GKZ system~\cite[Section~3.5]{Cruz19}, in which case the SST algorithm simplifies substantially. The conformal ideal $I_3(4,2,2,2)$ is a restriction of that GKZ system to a subspace of special values of some of the coefficients. Our further investigations will not rely on that fact; we here provide a general implementation of the SST algorithm.
\end{remark}

\subsection{Solutions to the indicial ideal}\label{sec:Groebnerw1}
We will compute the initial and indicial ideal of the $D$-ideal $I_3$ from \eqref{eq:indicial} with respect to the weight vector $w = (-1, 0, 1)\in C_1$. In the basis $v_1=(1,0,-1),$ $v_2=(-1,2,-1)$ of \Cref{sec:Grfan}, $w$ is $(-1,0)$. From \Cref{fig:grobnerfan}, we read that $w$ is contained in the relative interior of the cone~$C_1$ in the  Gr\"{o}bner fan, hence $w$ is generic for~$I_3$.
The initial ideal $\init_{(-w,w)}(I_3)$ is generated by the three~operators
\begin{align}\label{eq:initialwI}
\begin{split}
x_1\dx_1 + x_2\dx_2 + x_3\dx_3 + 1 \, , \quad
x_2\dx_2^2 + \dx_2 \, ,\quad
x_3\dx_3^2 + \dx_3 \, .
\end{split}
\end{align}
We compute the indicial ideal
\begin{align}
\ind_w(I_3) \,=\, R_3 \cdot \init_{(-w, w)}(I_3) \cap \CC [\theta_1, \theta_2, \theta_3] \, .
\end{align}
Users of {\sc Singular} may compute the initial ideal and a Gr\"{o}bner basis with respect to the weight vector $w$ with the following code.
{\small 
\begin{verbatim}
LIB "dmod.lib";
int c0 = 4; int c1 = 2; int c2 = 2; int c3 = 2;
ring r = 0,(x1,x2,x3,d1,d2,d3),dp; setring r;
def D3 = Weyl(r); setring D3;
poly P1 = 4*x1*d1^2 - 4*x3*d3^2 + (2*c0-4*c1+4)*d1 + (-2*c0+4*c3-4)*d3;
poly P2 = 4*x2*d2^2 - 4*x3*d3^2 + (2*c0-4*c2+4)*d2 + (-2*c0+4*c3-4)*d3;
poly P3 = 2*x1*d1+2*x2*d2 + 2*x3*d3 + (2*c0-c1-c2-c3);
ideal I = P1,P2,P3;
intvec w = -1,0,1; GBWeight(I,-w,w);
ideal inwI = initialIdealW(I,-w,w); inwI;
\end{verbatim}
}
This returns the initial ideal as in \eqref{eq:initialwI}. We deduce that the indicial ideal $\ind_{w}(I_3)$ is generated by the three operators
\begin{align}\label{eq:indwI}
\theta_1 + \theta_2 + \theta_3 + 1 \, , \quad
\theta_2^2\, , \quad
\theta_3^2 \, .
\end{align}
They are a Gr\"{o}bner basis of $\init_{(-w,w)}(I_3)$, and in this case $\init_{(-w,w)}(I_3)=\ind_w(I_3)$. One reads~that 
\begin{align}
V(\ind_w(I_3)) \,=\, \{ (-1,0,0) \},
\end{align}
i.e., the only zero of the indicial ideal is the point $(-1,0,0)$.
For the next step of the algorithm, we need to compute solutions to $\ind_w(I_3)$. These will become the starting monomials of the canonical series solutions to~$I$.
Recall that $\ind_w(I_3)=DJ$ for $J$ the $\CC[\theta_1,\theta_2,\theta_3]$-ideal generated by the operators in \eqref{eq:indwI}. We need to compute the primary decomposition of~$J$. In {\sc Singular}, one obtains the primary decomposition with the library {\tt primdec} by running the following code, where we encoded the Euler operators $\theta_1,\theta_2,\theta_3$ as variables {\tt th1},~{\tt th2},~{\tt th3}.
{\small 
\begin{verbatim}
LIB "primdec.lib";
ring r = 0,(th1,th2,th3),dp;
ideal J = th1 + th2 + th3 + 1, th2^2, th3^2;
list pr = primdecGTZ(J); pr;
\end{verbatim}
}
\noindent From the output, one reads the ideals
\begin{align}\label{eq:primaryI}
 \langle \theta_2^2\rangle \, , \quad \langle \theta_3^2 \rangle \, , \quad \langle \theta_1+\theta_2+\theta_3+1\rangle \, ,
\end{align}
coinciding with the Gr\"{o}bner basis in \eqref{eq:indwI}.
Recalling that $V(J)=\{A\}$ with $A=(-1,0,0)$, we read that $Q_A$ from \eqref{eq:primdecFrob} is the ideal generated by $\theta_1+\theta_2+\theta_3,$ $\theta_2^2,$ and $\theta_3^2$.
The orthogonal complement $Q_A^\perp$ from \eqref{eq:orthcompl} hence is
\begin{align}
Q_A^\perp \,=\, \left\{ p \in \CC[x_1,x_2,x_3] \, | \, \partial_2^2\bullet p = \partial_3^2\bullet p = (\partial_1+\partial_2+\partial_3)\bullet p \, = \, 0  \right\} \, .
\end{align}
This is a finite-dimensional vector space spanned by the $4$ polynomials
\begin{align}
\{1,\ x_1-x_2, \ x_1-x_3,\ (x_1-x_2)(x_1-x_3)\} \, .
\end{align}
Hence, by \Cref{{prop:solsind}}, the solution space to $\ind_w(I_3)$ is spanned by the $4$ functions
\begin{align}\label{eq:solsindw}
\begin{split}
m_1(x_1,x_2,x_3) &\,=\, x_1^{-1} \, , \\
m_2(x_1,x_2,x_3) &\,=\,x_1^{-1}\log\left(\frac{x_1}{x_2}\right)  , \\  
m_3(x_1,x_2,x_3) &\,=\,x_1^{-1}\log\left(\frac{x_1}{x_3}\right)  , \\
m_4(x_1,x_2,x_3) &\,=\,x_1^{-1} \log\left(\frac{x_1}{x_2}\right)\log\left(\frac{x_1}{x_3}\right)  .
\end{split}
\end{align}

In {\tt Macaulay2}~\cite{M2}, the generators of the indicial ideal in 
\eqref{eq:indwI} and its solutions are obtained conveniently using the commands {\tt distraction} and {\tt solvePDE} as follows.
{\small 
\begin{verbatim}
needsPackage "Dmodules"; needsPackage "NoetherianOperators";
D3 = QQ[x1,x2,x3,d1,d2,d3,WeylAlgebra=>{x1=>d1,x2=>d2,x3=>d3}]
c0 = 4; c1 = 2; c2 = 2; c3 = 2;
P1 = 4*x1*d1^2 - 4*x3*d3^2 + (2*c0-4*c1+4)*d1 + (-2*c0+4*c3-4)*d3;
P2 = 4*x2*d2^2 - 4*x3*d3^2 + (2*c0-4*c2+4)*d2 + (-2*c0+4*c3-4)*d3;
P3 = 2*x1*d1+2*x2*d2 + 2*x3*d3 + (2*c0-c1-c2-c3);
I = ideal(P1,P2,P3);
inwI = inw(I,{1,0,-1,-1,0,1}); indwI = distraction(inwI,QQ[th1,th2,th3])
solvePDE(indwI)
\end{verbatim}
}

\noindent Running this code outputs the following:
{\small 
\begin{verbatim}
      {{ideal (th3, th2, th1 + 1), {| 1 |, | dth1-dth2 |, | dth1-dth3 |, 
           | dth1^2-dth1dth2-dth1dth3+dth2dth3 |}}}
\end{verbatim}
}
\noindent As is explained in \cite{solvePDE}, the solutions to $\ind_w(I_3)$ are then obtained by replacing {\tt dthi} by $\log(x_i)$ and multiplying the resulting functions by $x^A=x_1^{-1}.$ 

\medskip

We will take advantage of the homogeneity of the system to write it in fewer variables. Consider the change~of~variables 
\begin{align}\label{eq:changecord}
y_1 = x_1\,,\quad y_2 = \frac{x_2}{x_1}\,, \quad y_3 = \frac{x_3}{x_1} \,.
\end{align}
\begin{lemma}\label{lem:twovarprop}
Every solution $f(x_1,x_2,x_3)$ to $I_3$ can be written in the form 
\begin{align}
    f(x_1,x_2,x_3) \,=\, x_1^{-1}\tilde{f}(y_2, y_3) \,,
\end{align}
where $\tilde{f}(y_2, y_3)$ denotes the function $\tilde{f}(y_2,y_3)= f(1,y_2,y_3).$
\end{lemma}
\begin{proof}
Since $I\bullet f(x_1, x_2, x_3) = 0$, $f$ is a solution of  $\theta_1 + \theta_2 + \theta_3 + 1$, in particular. Thus, by the converse of Euler's homogeneous function theorem, $f$ is homogeneous of degree $-1$. 
This yields $f(1, y_2, y_3)= x_1 f(x_1, x_2, x_3)$. Setting $\tilde{f}(y_2, y_3) \coloneqq f(1, y_2, y_3)$ proves the claim. 
\end{proof}
Therefore, we can write each canonical series in the form
\begin{equation}
x_1^{-1} \, \hspace*{-2mm} \sum_{m, n \in \mathbb{Z}} \sum_{0 \leq i, j \leq 3} c_{mnij} \, y_2^m y_3^n(\log y_2)^i(\log y_3)^j.
\end{equation}

We will call an element of the form $(\log y_2)^i(\log y_3)^j$ a \emph{monomial}.
We now encode our $D$-ideal in the new variables.
Now $D$
denotes the Weyl algebra in the $y$-variables, i.e., $D=\CC[y_1,y_2,y_3]\langle \dy{1},\dy{2},\dy{3}\rangle$, where $\dy{i}=\partial/\partial y_i.$
Denote $\thy{i} = y_i \dy{i}.$ Then
\begin{align}\label{eq:thetaxy}
\begin{split}
\theta_1=\thy{1}-\thy{2}-\thy{3}, \quad \th_2 = \thy{2}, \quad \th_3 = \thy{3} \,.
\end{split} \end{align}

\noindent In the $y$-variables, we thus obtain the following basis of the solutions space of $\ind_w(I_3)$:
\begin{align} \label{eq:indsol1}
\begin{split}
m_1(y_1,y_2,y_3) &\,=\, y_1^{-1} \, ,\\
m_2(y_1,y_2,y_3) & \,=\,  y_1^{-1}\log(y_2) \, ,\\
m_3(y_1,y_2,y_3) &\,=\, y_1^{-1}\log(y_3) \, ,\\
m_4(y_1,y_2,y_3) &\,=\, y_1^{-1}\log(y_2)\log(y_3) \, . 
\end{split}
\end{align}
After substituting back to the $x_i$ variables, these functions are---up to reordering and a sign---exactly the functions found in~\eqref{eq:solsindw}.
Since each $m_i$ is entirely contained within a single~$L_p,$ we have $\text{Start}_{\prec_w}(m_i) = m_i$ (cf.~\cite[Lemma 2.5.10]{SST00}).

\subsection{Lifting the initial monomials to solutions of \texorpdfstring{$I_3$}{I\_3}}
We proceed using the algorithm described in \Cref{thm:liftalgo}. The space $L_p$ is 16-dimensional and defined as the span of elements of the form $x_1^{-1}y^p(\log y_2)^m (\log y_3)^n$ where $p \in \ZZ^2$ is fixed and $ 0 \leq m, n \leq 3.$  In order to execute the algorithm, we must first find the matrix of the transformation from $L_p'$ to $(L_p)^d.$ As we will see below, this matrix consists of $3$ vertically concatenated blocks, where each block is an upper triangular $16\times 16$ matrix.

\noindent First, we compute that a Gr\"{o}bner basis $G_w$ of $I_3$ with respect to~$w$ is given by
\begin{align}\label{eq:GrIw}\begin{split}
g_1 & \,=\, (\theta_2+\theta_3+1)\partial_1+(\theta_2+1)\partial_2 \, , \\g_2 & \,=\, (\theta_2+\theta_3+1)\partial_1+(\theta_3+1)\partial_3 \, , \\
g_3 & \,=\, \theta_1+\theta_2+\theta_3+1 \, .
\end{split}
\end{align}
Switching to the $y$-variables, the operators in \eqref{eq:GrIw} turn into
\begin{align}\label{eq:g1g3g3y}
\begin{split}
    y_1y_2g_1 &\,=\, \thy{2}^2 +y_2(\thy{1}\thy{2}+\thy{1}\thy{3}-\thy{2}^2-\thy{2}-2\thy{2}\thy{3} -\thy{3}^2-\thy{3}+\thy{1}) \, , \\
    y_1y_3g_2 & \,=\, \thy{3}^2 + y_3( \thy{1}\thy{2}+\thy{1}\thy{3}-\thy{2}^2-\thy2-2\thy{2}\thy{3}-\thy{3}^2-\thy{3}+\thy{1}) \, , \\
    g_3 & \,=\, \thy{1} + 1 \, .
\end{split}
\end{align}
Replacing $\thy{1}$ by $-1$ gives
\begin{align}\label{eq:g1g2fh}
\begin{split}
    y_1y_2g_1 &\,=\, \thy{2}^2 - y_2(\thy{2} + \thy{3} + 1)^2 \, , \\
    y_1y_3g_2 & \,=\, \thy{3}^2 - y_3(\thy{2} + \thy{3} + 1)^2\, .
\end{split}
\end{align}
Observe that $y_1 g_1=Q_2-Q_1$ and $y_1 g_2=-Q_1$ for $Q_1,Q_2$ from \eqref{eq:I3y} from \Cref{sec:Pfaffian}.

\begin{remark}
Above, we pass to the $D$-module $M=D/I_3$. In $M$, $\thy{1}\equiv -1$. This justifies to substitute $\thy{1}$ by $-1$. We point out that the ideals generated by the operators in \eqref{eq:g1g3g3y} and \eqref{eq:g1g2fh}, respectively, are not the same ideals; but the quotients by them are isomorphic as $D$-modules on the algebraic torus, and hence their solution spaces are isomorphic. 
\end{remark} 

We now check that the order condition on the $h_i$'s is satisfied. Indeed, for $ w = (-1,0,1),$ the variable $y_2 = x_2/x_1$ has $w$-weight $-w \cdot (-1, 1, 0) = -1.$ Similarly, $y_3$ has $w$-weight $-2.$ Thus $\text{ord}_{(-w, w)}(h_i) = \text{ord}_{(-w, w)}(y_i)<0$ for $i = 2, 3,$ as desired. In other words, $g_1$ gives rise to a recurrence relation between $L_p$ and $L_{p + (1, 0)},$ while $g_2$ gives a recurrence relation between the spaces $L_p$ and $L_{p + (0, 1)}.$ 

\medskip

We have implemented the SST algorithm from the proof of \Cref{thm:liftalgo} in {\tt Sage} for PDEs in two variables. Running it produces the canonical series solutions that are shown in~\Cref{eq:canSeriesC1} at the very end of this section.

\medskip

In summary, we achieved a basis of holomorphic solutions to $I_3$ purely from the Gr\"{o}bner data of $I_3$. We carried the computations out for a weight vector $w$ from the cone $C_1$ of the Gr\"{o}bner fan of $I_3$. For our computations, we changed to $y$-variables. To refine the weight-ordering of $w$ to a term ordering on the Weyl algebra, we used the degree reverse lexicographical ordering for $x_1>x_2>x_3> \partial_1>\partial_2 > \partial_3$.
In order to compute a basis of the solution space to $I_3$, one could equally have started from cone~$C_2$ or~$C_3$, resulting in different orderings. One would change the refining order to $x_2 > x_1, x_3 > \partial_2 > \partial_1, \partial_3$ for~$C_2$, and to $ x_3 > x_1, x_2 > \partial_3 > \partial_1, \partial_2 $ for~$C_3$. The Gr\"{o}bner bases of~$I_3$ w.r.t.\ weights from $C_2$ and $C_3$ are then obtained by a relabelling of the indices, and accordingly for the starting monomials and canonical series solutions. In \Cref{lem:twovarprop}, one would factor out $x_2^{-1}$ for $C_2$, and $x_3^{-1}$ for~$C_3$. 

\bigskip
We end this section by giving the canonical series solution of $I_3$ for $w$ from~$C_1$ obtained by our implementation of the SST algorithm. 
We display them from $w$-weight $0$ to $4$ on the next page.\pagebreak
{\small
\begin{align}\label{eq:canSeriesC1}
\begin{split}
\tilde{f}_1(y_2, y_3) &\,=\, 1 + y_2 + y_3 + y_2^{2} + 4y_2y_3 + y_3^{2} + y_2^{3} + 9y_2^{2}y_3 + y_2^{4} \,+\, 
 \cdots \, , \\
\tilde{f}_2(y_2, y_3) &\,=\, \log(y_2) + \log(y_2)y_2 + (2 + \log(y_2))y_3 + \log(y_2)y_2^{2}  + (4 + 4\log(y_2))y_2y_3  \\
& \quad  \ + (3 + \log(y_2))y_3^{2} + (\log(y_2))y_2^{3} + (6 + 9\log(y_2))y_2^{2}y_3 + \log(y_2) y_2^{4} \,+\,  \cdots \, , \\
\tilde{f}_3(y_2, y_3) &\,=\, \log(y_3) + (2 + \log(y_3))y_2 + \log(y_3)y_3 + (3 + \log(y_3))y_2^{2} \\
  & \quad \ + (4 + 4\log(y_3))y_2y_3 + 
 \log(y_3)y_3^{2} + \left(\frac{11}{3} + \log(y_3)\right)y_2^{3} \\
  & \quad \ + (15 + 9\log(y_3))y_2^{2}y_3 + \left(\frac{25}{6} + \log(y_3)\right)y_2^{4}  \,+\,  \cdots \, ,\\
\tilde{f}_4(y_2, y_3) &\,=\, \log(y_2)\log(y_3) + (-2 + 2\log(y_2) + \log(y_2)\log(y_3))y_2 \\
&\quad  \ + \left(-2 + 2\log(y_3) + \log(y_2)\log(y_3)\right)y_3
 \\ & \quad \
+  \left(-\frac{5}{2} + 3\log(y_2) + 
 \log(y_2)\log(y_3) \right)y_2^{2} \\
& \quad \ + \left(-6 + 4\log(y_2) + 4\log(y_3) + 4\log(y_2)\log(y_3)\right)y_2y_3  \\
 & \quad \ + \left(-\frac{5}{2} + 3\log(y_3) + \log(y_2)\log(y_3)\right)y_3^{2} 
 \\ & \quad \  + \left(-\frac{49}{18} + \frac{11}{3}\log(y_2) + 
 \log(y_2)\log(y_3)\right)y_2^{3}  \\
  & \quad \   + \left(-\frac{29}{2} + 15\log(y_2) + 6\log(y_3) +
 9\log(y_2)\log(y_3)\right)y_2^{2}y_3  \\
 & \quad \ + \left(-\frac{205}{72} + \frac{25}{6}\log(y_2) + \log(y_2)\log(y_3)\right)y_2^{4} \,+\, \cdots \, . \\
\end{split}\end{align}
}

\noindent We recall that these four functions are a basis of the 
space of solutions to the \mbox{$D$-ideal}~$I_3$. In order to pick out a specific solution, such as the one corresponding to the triangle Feynman integral from which we started, one needs to prescribe $\operatorname{rank}(I_3)$-many initial conditions. This can be done in a number of ways. 
For example, we may numerically evaluate the triangle Feynman integral $J_{4;1,1,1}^{\rm triangle}$ at four arbitrary kinematic points.\footnote{The triangle Feynman integral as well as the known solutions to $I_3$ in~\eqref{eq:f1234} are multivalued functions. We consider them in the so-called ``Euclidean region'', i.e., we assume that $x_1,x_2, x_3<0\, .$}
For each cone, the kinematic points $x=(x_1,x_2,x_3)$ must be chosen so that consecutive truncations of the canonical series evaluated at $x$ converge. 
In this way, the truncation of the canonical series at \mbox{$w$-weight} $k$ provides an approximation of the solutions, with an error comparable to the size of the omitted terms of $w$-weight $k+1$.
For example, in the cone $C_1$ with weight vector $w=(-1,0,1)$, we must have that $|x_1| \gg |x_2| \gg |x_3|$.
Comparing the numerical evaluations obtained using \textsc{AMFlow}~\cite{Liu:2022chg} to a linear combination of the canonical series in~\eqref{eq:canSeriesC1} allows us to determine that
\begin{align}
J_{4;1,1,1}^{\rm triangle}(x_1,x_2,x_3) \,=\, x_1^{-1} \biggl[ -\tilde{f}_4\left(\frac{x_2}{x_1},\frac{x_3}{x_1} \right) - \frac{\pi^2}{3} \tilde{f}_1\left(\frac{x_2}{x_1},\frac{x_3}{x_1} \right) \biggr]\,.
\end{align}
We validated this against further numerical evaluations, as well as against the analytic solution in~\eqref{eq:triangle}.

\section{Series solutions using Wasow's method}\label{sec:solWasow}
In this section, we discuss how to construct series solutions of the ideal $I_3$ using the toolkit developed for computing Feynman integrals in high energy physics. We provide an alternative approach for constructing the canonical series solutions discussed in the previous sections without resorting to Gr\"{o}bner basis techniques, and give a dictionary between the two methods. 
We also touch base with the problem of solving a first-order Pfaffian system of PDEs, namely a system of PDEs of the form
\begin{align} \label{eq:dF}
\frac{\partial}{\partial x_i} \bullet F(x) \,=\,  P_{i}(x) \cdot F(x) \qquad \text {for } \ i \,=\, 1,\ldots,n  \,,
\end{align}
where $F(x)$ is a vector-valued function of $x=(x_1,\ldots,x_n)$. This is a typical problem that arises in the evaluation of Feynman integrals. Often, one aims to approximate the solution $F(x)$ for small values of some variable, say of~$x_1$. 
If the Pfaffian system~\eqref{eq:dF} is in a \textit{manifestly Fuchsian form at $x_1=0$}, i.e., if the matrices $P_i(x)$ have at most simple poles at $x_1=0$, there is an algorithm~\cite{Wasow} to construct the asymptotic series of the solution $F(x)$ around $x_1=0$; see e.g.~\cite{Zoia:2021zmb} for a discussion in the context of Feynman integrals. They will be of the form
\begin{align} \label{eq:asymexp_example}
F(x) \,=\, \sum_{k\ge 0} \sum_{m=0}^{m_{\text{max}}} c_{k,m}(x_2,\ldots, x_n) \, x_1^k \log(x_1)^m \,,
\end{align} 
up to an arbitrary power of $x_1$, with $m_{\text{max}}$ a positive integer which depends on the specific Pfaffian system, and $c_{k,m}$ are Laurent series in $x_2,\ldots,x_n$. 
The series in \eqref{eq:asymexp_example} is a particular case of \eqref{eq:Fkintro}: for the  weight $w=(1,0,\ldots,0)$, the series in \eqref{eq:Fkintro} is the truncation of~\eqref{eq:asymexp_example}.

A number of issues needs to be addressed in order to apply this method to $D$-ideals. First of all, we need to construct a Pfaffian system associated with the considered \mbox{$D$-ideal}. In \cite[p.~23]{SatStu19}, it is explained how to achieve this by Gr\"obner basis computations. In \Cref{sec:Pfaffian}, we present an alternative algorithm which relies on linear algebra only, and apply it to the $D$-ideal $I_3$. This method is inspired by the problem of integration by parts (IBP) reduction for Feynman integrals~\cite{Chetyrkin:1981qh,Tkachov:1981wb}. We spell this analogy out, and provide a dictionary of the relevant concepts. As a by-product, this method allows us to determine the holonomic rank and the singular locus of the $D$-ideal without computing Gr\"{o}bner bases.
The resulting Pfaffian system is in general not in a manifestly Fuchsian form; double or higher-order poles may appear in the DEs, preventing us from applying the algorithm from~\cite{Wasow} to construct the asymptotic series of the solutions. Finding a gauge transformation which puts the Pfaffian system into Fuchsian form is a well known problem in Feynman integrals~\cite{henn2013multiloop}, and a number of strategies 
 have been developed in that context, see e.g.~\cite{Henn:2014qga,Lee:2014ioa}. In \Cref{sec:Fuchsian}, we argue that, if a $D$-ideal is regular holonomic, 
for every singular point $x^*$ it is possible to find a gauge transformation such that the associated Pfaffian system has at most simple poles at~$x=x^*$. In the case of~$I_3$, we show that we can even construct a Pfaffian system which is manifestly Fuchsian everywhere, namely which has at most simple poles at any singular point. The resulting system is particularly simple; we can solve it analytically, and by doing so, we will recover the solutions to~$I_3$ given in~\eqref{eq:f1234}.
Finally, we need to fill an important conceptual gap. The canonical series computed by Gr\"obner bases are truncated by the $w$-weight of the terms. The asymptotic series expansions, on the other hand, rely on the notion of a ``small parameter'', e.g.\ $x_1$ in \Cref{eq:asymexp_example}, and are truncated at a given power of it. To link the two notions, in \Cref{sec:Asymptotic}, we introduce an auxiliary variable~$t$, which captures the $w$-weight of the monomials. We compute the asymptotic solutions of the Pfaffian system around $t=0$, and verify that they coincide with the canonical series computed in \Cref{sec:solI3Groebner}. 

\subsection{Pfaffian systems}\label{sec:Pfaffian} 
In this section, we present a method for constructing a Pfaffian system of PDEs associated with a given $D$-ideal using linear algebra. We are going to apply that method to the conformal $D$-ideal $I_3$. A closely related method to efficiently construct Pfaffian matrices is discussed in~\cite{MacaulayFeynman}. Our presentation builds on the analogy with IBP reduction, to the benefit of the readers who are familiar with Feynman integrals.

In \Cref{sec:intPfaffian}, we have seen that the vector-valued function $F(x)$, which satisfies a Pfaffian system~\eqref{eq:dF} associated with a $D_n$-ideal $I$ with $\rank(I)=m$ is given by
\begin{align} \label{eq:Fgendef}
F(x) \,=\, \begin{bmatrix} \partial^{p(1)} \bullet f(x) \\ \vdots \\ \partial^{p(m)} \bullet f(x) \end{bmatrix} \,,
\end{align}
where $f(x)$ is a general holomorphic solution of $I$, and $p(1),\ldots,p(m)\in \NN^n$ such that $\{\partial^{p(1)}, \ldots, \partial^{p(m)}\}$ are the standard monomials of the left  $R$-ideal $RI$ for a given monomial ordering. We may assume $\partial^{p(1)} = 1$, so that the first entry of $F(x)$ is a holomorphic solution of~$I$. The standard monomials can be determined by Gr\"obner bases in $R$, and the matrices $P_i(x)$ in the Pfaffian system~\eqref{eq:dF} can be obtained through 
reduction in $R$,
\begin{align}
\partial_i  \partial^{p(j)} \,=\, \sum_{k=1}^m \left(P_i(x)\right)_{jk} \partial^{p(k)} \ \text{mod} \, RI \,.
\end{align}

Our goal is to construct a Pfaffian system of $I$ using linear algebra only. Let $Q_1,\ldots,Q_N$ be generators of $I$. First of all, the $R$-monomials $\partial^{p(i)}$ defining $F(x)$ as in~\eqref{eq:Fgendef} need not to be standard monomials for $F(x)$ to satisfy a first-order Pfaffian system of PDEs. It suffices that they are a $\CC(x)$-basis of $R/RI$. One way to find such a basis is to write down the relations among the monomials $\partial^a$ in $R/RI$, and solve them. There are infinitely many such relations in $R/RI$, obtained by multiplying the generators $Q_i$ of $I$ by any $\partial^a$, 
\begin{align} \label{eq:daQi}
\left\{ \partial^a Q_i  \equiv 0 \ \operatorname{mod}\, RI  \, | \,  i=1,\ldots,N, \ a \in \mathbb{N}_0^n \right\} \,.
\end{align}
We view this as a linear system of equations in the ``variables'' $\partial^a$. In \Cref{sec:intPfaffian}, we have seen that, if $I$ has finite holonomic rank $m$, the linear system~\eqref{eq:daQi} has $m$ independent variables. In other words, any $\partial^a$ can be expressed as a linear combination of $m$ $\partial^a$'s in $R/RI$ by solving the linear system~\eqref{eq:daQi}. The issue is that there are infinitely many equations. We cannot solve them all, nor do we need to. We instead adopt an approach inspired by Laporta's algorithm~\cite{Laporta:2000dsw} for solving integration-by-parts relations in the context of Feynman integrals.

We introduce an ordering $>$ of the unknowns $\partial^a$ graded by the total degree of the derivatives. The ordering among unknowns with the same degree is instead arbitrary. It will lead to a different basis of $R/RI$.
For instance, we choose the graded lexicographic order. Explicitly, for $n=2$, we have
\begin{align}
\cdots \,>\, \pt_1^2 \,>\, \pt_1 \pt_2 \,>\, \pt_2^2 \,>\, \pt_1 \,>\, \pt_2 \,>\, 1\, .
\end{align}
Next, we multiply
the generators $Q_i\in D$ by $\partial$'s from the left, up to a certain degree $d_{\text{max}}$.\footnote{If the generators $Q_i$ have very different degrees in the various variables, it may be more efficient to seed up to a different degree for each variable.} This yields the equations 
\begin{align}
 \partial^a Q_i  \equiv 0 \quad \text{mod} \, RI\,,  \quad i = 1,\ldots, N , 
\end{align} 
where $a\in \NN^n$ runs over all vectors of natural numbers with $ a_1 + \cdots + a_n \le d_{\text{max}}$. We denote the resulting finite system of equations by $S_{d_{\text{max}}}$.
We dub this operation \textit{seeding}, in analogy with the equivalent step in constructing IBP relations for Feynman integrals.

Since we have truncated the system of equations, some unknowns cannot be solved for with the equations in $S_{d_{\text{max}}}$. The key idea is to solve them only for a subset of the unknowns, in particular those with total degree up to some value $d_{\text{needed}}$ below $d_{\text{max}}$ (e.g.\ $d_{\text{needed}}=d_{\text{max}}-1$), and to eliminate the other unknowns through Gaussian elimination. We denote by $\mathbb{D}$ the vector containing all unknowns $\partial^a$ appearing in $S_{d_{\text{max}}}$ sorted according to the ordering defined above, and by $\mathbb{D}_{\text{needed}}$ the needed unknowns,
\begin{align}
\mathbb{D}_{\text{needed}} \,=\, \left\{ \partial^a \in \mathbb{D} \,|\, a_1+\cdots+a_n \le d_{\text{needed}} \right\} \,.
\end{align} 
By construction, $\mathbb{D} $ has the block form
\begin{align}
\mathbb{D} \,=\, \begin{bmatrix} \mathbb{D}_{\text{extra}}  \\ \hline \mathbb{D}_{\text{needed}} \\ \end{bmatrix} \,,
\end{align}
with $ \mathbb{D}_{\text{extra}}  = \mathbb{D}  \backslash \mathbb{D}_{\text{needed}} $.
We represent the equations in $S_{d_{\text{max}}}$ in matrix form as
\begin{align}
M(x) \cdot \mathbb{D} \,\equiv\, 0 \quad \text{mod} \, RI \,,
\end{align}
where $M(x)$ is a $[S_{d_{\text{max}}}] \times [\mathbb{D}]$ matrix whose entries depend on $x$. 
We row-reduce $M(x)$ and drop the null rows. The resulting system of equations has the form
\begin{align}
\left[\begin{array}{c|c} A(x)  & B(x) \\ \hline 0  & C(x) \end{array} \right] \cdot \begin{bmatrix}  \mathbb{D}_{\text{extra}}  \\ \hline \mathbb{D}_{\text{needed}} \\ \end{bmatrix} \,\equiv\, 0 \quad \text{mod} \, RI \,.
\end{align}
Finally, we perform back substitution, but only for the needed unknowns. In other words, we solve the subset of equations
\begin{align}
C(x) \cdot  \mathbb{D}_{\text{needed}} \, \equiv \, 0 \quad \text{mod} \, RI \,.
\end{align}
As a result, we obtain the expressions of the needed unknowns in terms of a subset of independent monomials $\partial^a$ in $R/RI$. Let us make a few comments about this algorithm.
\begin{enumerate}[(1)]
\item The value of $d_{\text{max}}$ must be higher than the maximal degree of the standard monomials. In practice, we iterate the algorithm for various increasing values of $d_{\text{max}}$ until we see that the resulting independent monomials stay the same. These are a basis of $R/RI$.
\item If the $D$-ideal $I$ has infinite holonomic rank, the algorithm will not terminate. In practice, as we increase $d_{\text{max}}$, we will obtain more and more independent monomials. Thus, this procedure is a linear algebra approach for computing the holonomic rank of a $D$-ideal.
\item Let $\{\pt^{p(1)},\ldots ,\pt^{p(m)} \}$ be a $\CC(x)$-basis of $R/RI$. If the needed unknowns include $\pt_i \pt^{p(j)}$ for any $i=1,\ldots,n$ and $j=1,\ldots,m$, the matrices $P_i(x)$ of the Pfaffian system~\eqref{eq:dF} can be read off from the solution of the system of equations discussed above.
\item This algorithm does not prove that the identified monomials $\partial^a$ are a basis of $R/RI$. However, the resulting Pfaffian system of DEs must satisfy the integrability conditions~\eqref{eq:intmatrices}. If the latter are not satisfied, either a higher value of $d_{\text{max}}$ must be used, or the ideal has infinite holonomic rank. If they are satisfied, the solution space has a $\mathbb{C}$-dimension equal to the length of $F(x)$ (see \Cref{sec:intPfaffian}). Since the first component of $F(x)$ is by construction a general solution of the original ideal $I$, we can conclude that the algorithm has terminated, and that a basis of $R/RI$ has been identified.
\item The common denominator of the entries of the matrices $P_i$ in the Pfaffian system~\eqref{eq:dF} defines a hypersurface which contains the singular locus of $I$.
\item There is a strong analogy with IBP reduction in Feynman integrals, which we summarized in \Cref{tab:analogy}.
\begin{table}[t!]
\begin{center}
\begin{tabular}{|c|c|}
\hline
Construction of a Pfaffian system & IBP reduction with Laporta's algorithm \\
\hline
$\partial^a $ & Feynman integrals \\
$a$ in $\partial^a$  & propagator powers \\
$\partial^a Q_i = 0$ in $ R/RI$ & IBP identities \\
$\CC(x)$-basis of $R/RI$ & a set of master integrals \\
\hline
\end{tabular}
\end{center}
\caption{Analogies between the construction of a first-order Pfaffian system of DEs associated with a given $D$-ideal and IBP reduction of Feynman integrals.}
\label{tab:analogy}
\end{table}
\end{enumerate}

\smallskip

We now apply this algorithm to the conformal ideal $I_3$. It is convenient to use the variables $y$ defined in \eqref{eq:changecord}, and to eliminate $y_1$ from the problem as follows. \Cref{lem:twovarprop} guarantees that every solution of $I_3$ has the form $\tilde{f}(y_2,y_3)/y_1$. Any function of such form is annihilated by the generator $\tilde{P}_3$ in \eqref{eq:p1p2p3tilde}. The action of the other two generators can be expressed as follows:
\begin{align}
x_1^2 \, \tilde{P}_i \bullet \left[\frac{1}{x_1} 
\tilde{f}\left(\frac{x_2}{x_1},\frac{x_3}{x_1}\right) \right] \,=\, Q_i \bullet \tilde{f}(y_2, y_3) \quad \text{for }\ i\,=\,1,2\, .
\end{align}
 Here, $Q_1$ and $Q_2$ are differential operators which depend on $y_2$ and $y_3$ only:
\begin{align}\begin{split}\label{eq:I3y}
Q_1 & \,=\, y_2^2 \pt_{y_2}^2 + 2 y_2 y_3 \pt_{y_2} \pt_{y_3} + (y_3-1) y_3 \pt_{y_3}^2  + 3 y_2 \pt_{y_2} + (3 y_3-1) \pt_{y_3} + 1 \,, \\
Q_2 &\,=\, y_2 \pt_{y_2}^2 - y_3 \pt_{y_3}^2  + \pt_{y_2} - \pt_{y_3}\,.
\end{split}\end{align}
We denote by $I_3^y$ the $D$-ideal generated by $Q_1$ and $Q_2$. It is holonomic, it has holonomic rank $4$, and its singular locus is
\begin{align} \label{eq:singI3y}
\Sing \left( I_3^y \right) \,=\, V(y_2 y_3 \tilde{\lambda}) \,,
\end{align}
where $\tilde{\lambda}$ is obtained from the polynomial $\lambda$~\eqref{eq:lambda} as
\begin{align}
\tilde{\lambda} \,=\, \lambda|_{x_1=1, \, x_2=y_2, \, x_3=y_3} \, .
\end{align}

The general solution of $I_3$ is given by the general solution of $I_3^y$ divided by $y_1$. Hereinafter, we will thus focus on $I_3^y$. Since the latter depends on $y_2$ and $y_3$ only, we use the shorthand notations $y = (y_2,y_3)$, $y^a = y_2^{a_2} y_3^{a_3}$, and $\partial_y^a = \dy{2}^{a_2} \dy{3}^{a_3}$.
We compute a basis of $R/RI_3^y$ by running the algorithm described above. We use $d_{\text{needed}}=d_{\text{max}}-1$. For $d_{\text{max}}=1$, we find only one basis element, namely $1$. Starting from $d_{\text{max}}=2$ and onwards, we identify four basis elements. We define 
\begin{align} \label{eq:Fdef}
\tilde{F}(y) \, \coloneqq \, \begin{bmatrix} \tilde{f}(y) \\ \pt_{y_2} \bullet \tilde{f}(y) \\ \pt_{y_3} \bullet \tilde{f}(y) \\ \pt_{y_3}^2 \bullet \tilde{f}(y) \end{bmatrix}\,,
\end{align}
where $\tilde{f}(y)$ is a holomorphic solution of $I_3^y$. It satisfies a system of first-order Pfaffian DEs
\begin{align}
\pt_{y_i} \bullet \tilde{F}(y) \,=\, \tilde{P}_{i}(y) \cdot \tilde{F}(y) \,,
\end{align}
and the $4 \times 4$ matrices $\tilde{P}_{i}$ satisfy the integrability conditions. For example, $\tilde{P}_2$ is the~matrix 
\begin{align} \label{eq:A1}
\tilde{P}_2 \,=\, \begin{bmatrix}
0   & 0 &  \frac{-1}{2  y_2  y_3} &  \frac{\tilde{\lambda}-2 y_3 (1 + y_2 - y_3)}{2  y_2  y_3^2 \, \tilde{\lambda}}\\
1   & -\frac{1}{y_2} &  -\frac{1}{y_3}  & \frac{ \tilde{\lambda} - y_3 (1 + 3 y_2 - y_3) }{ y_3^2 \, \tilde{\lambda}} \\  
0 & \frac{1}{y_2} &  \frac{1 -  y_2 - 3  y_3}{2  y_2  y_3} &  \frac{ (y_2 + y_3 - 1) \left(\tilde{\lambda} - 2 y_3 (2 + 3 y_2 - 4 y_3) \right) }{2  y_2  y_3^2 \, \tilde{\lambda}}\\
 0 & \frac{y_3}{y_2} & \frac{1-y_2-y_3}{2 y_2} & \frac{3}{2 y_2} + \frac{y_2-1}{2 y_2 y_3} + {\scriptstyle 5} \frac{1 - y_2 + y_3}{\tilde{\lambda}}
\end{bmatrix}^{\top} \,.
\end{align} 
From \Cref{eq:A1}, we see that the common denominator is $y_2 y_3^2 \tilde{\lambda}$, which is in agreement with the singular locus computed via Gr\"obner bases. Furthermore, some entries have double poles at $y_3 = 0$. The system is thus not in Fuchsian form. 
In the next section, we will perform a gauge transformation to obtain a Pfaffian system which is in a manifestly Fuchsian form.

\subsection{Writing the system in Fuchsian form}\label{sec:Fuchsian}
If a $D$-ideal $I$ is regular holonomic, all solutions have moderate growth when approaching the singular locus (cf.~\Cref{sec:Weylalgebra}). This implies that the solutions cannot have essential singularities. This has implications for the form of the associated first-order Pfaffian system of DEs. 
A first-order ODE having poles of order higher than~$1$ implies that its solution has an essential singularity. The moderate growth condition on the solutions to regular holonomic $D$-ideals thus implies that the associated Pfaffian system has at most simple poles. 
The Pfaffian systems may however have ``spurious'' double poles. For example, the system
\begin{align} \label{eq:spurious_example}
\partial_x \bullet F(x) \,=\, \begin{bmatrix} -\frac{1}{x} & \frac{1}{x^2} \\ 0 & 0 \\ \end{bmatrix} \cdot F(x)
\end{align}
exhibits a double pole at $x=0$, yet the solution
\begin{align}
 F(x) \,=\, \begin{bmatrix}  c_1 \frac{\log(x)}{x} + \frac{c_2}{x} \\  c_1 \end{bmatrix}
\end{align}
has no essential singularity at $x=0$. Such spurious poles can be removed by a suitable gauge transformation,
\begin{align}
F(x) \,=\, T(x) \cdot G(x) \,.
\end{align}
The goal is to construct the transformation matrix $T(x)$ such that the new vector-valued function $G(x)$ satisfies a first-order system of DEs which is manifestly Fuchsian at $x=0$. For this simple rank-two example this can be achieved by 
\begin{align}
T(x) \,=\, \begin{bmatrix} \frac{1}{x} & 0 \\ 0 & 1 \end{bmatrix} \,,
\end{align}
which leads to
\begin{align}
\partial_x \bullet G(x) \,=\, \begin{bmatrix} 0 & \frac{1}{x} \\ 0 & 0 \\ \end{bmatrix} \cdot G(x) \,.
\end{align}

The problem of finding a gauge transformation which puts a Pfaffian system into a manifestly Fuchsian form is central in the modern methodology for computing Feynman integrals~\cite{henn2013multiloop}. 
For our purposes, it suffices to put the Pfaffian system in a form which is manifestly Fuchsian at the singular point around which we compute the asymptotic expansion. 
If the entries of the required transformation matrix are in the field of rational functions, this problem is solved (see~\cite{Henn:2014qga,Lee:2014ioa} and the references therein).
For the Pfaffian system associated with the $D$-ideal $I_3^y$ constructed in the previous subsection, we perform the gauge transformation 
\begin{align} \label{eq:G2F}
\tilde{F}(y) \,=\, \tilde{T}(y) \cdot \tilde{G}(y)
\end{align}
with the following transformation matrix:\footnote{This transformation matrix contains the square root $\sqrt{\tilde{\lambda}}$, which can be rationalized by changing variables from $y$ to $(z,\bar{z})$, defined through $y_2 = z \bar{z}$, $y_3 = (1-z)(1-\bar{z})$, so that $\tilde{\lambda} = (z-\bar{z})^2$.}
\begin{align} \label{eq:T}
\tilde{T}(y) \,=\, \begin{bmatrix}
\frac{1}{\sqrt{\tilde{\lambda}}} & 0 & 0 & 0 \\
\frac{1-y_2+y_3}{\tilde{\lambda}^{\frac{3}{2}}} & - \frac{1}{\tilde{\lambda}}  & \frac{y_2+y_3-1}{2 y_1 \tilde{\lambda}} & 0 \\
\frac{1+y_2-y_3}{\tilde{\lambda}^{\frac{3}{2}}} & \frac{y_2+y_3-1}{2 y_3 \tilde{\lambda}} & -\frac{1}{\tilde{\lambda}} & 0 \\
\frac{2(\tilde{\lambda} + 6 y_2)}{\tilde{\lambda}^{\frac{5}{2}}} & 
  \frac{6 y_2 y_3 (1 + y_3)  -6 y_3 (1-y_3)^2  + \tilde{\lambda} (1 - y_2 + 3 y_3) }{2 y_3^2 \tilde{\lambda}^2} & 
  \frac{3 (y_3-y_2-1)}{\tilde{\lambda}^2} & 
  \frac{-1}{y_3 \tilde{\lambda}} \\
\end{bmatrix} \,.
\end{align}
As a result, $\tilde{G}(y)$ satisfies a linear system of first-order DEs which is in 
Fuchsian form. It can be elegantly expressed as
\begin{align} \label{eq:fuchsian}
\d \tilde{G}(y) \,=\, \d \tilde{B}(y) \cdot \tilde{G}(y) \,,
\end{align} 
where
\begin{align}\label{eq:By}
\tilde{B}(y) \,=\, \begin{bmatrix}
0 & \frac{1}{2} \log\bigl(\tilde{l}_4\bigr) & -\frac{1}{2} \log\bigl(\tilde{l}_3\bigr)-\frac{1}{2}\log\bigl(\tilde{l}_4\bigr) & 0 \\
0 & 0 & 0 &  \log\bigl(\tilde{l}_1\bigr) \\
0 & 0 & 0 &  \log\bigl(\tilde{l}_2\bigr) \\
0 & 0 & 0 & 0 \\ 
\end{bmatrix} \,,
\end{align} 
and
\begin{align} \label{eq:alphabet}
\tilde{l}_1 \,=\, y_2 \,, \quad
\tilde{l}_2 \,=\, y_3 \,, \quad
\tilde{l}_3 \,=\, \frac{1-y_2-y_3 - \sqrt{\tilde{\lambda}}}{1-y_2-y_3+\sqrt{\tilde{\lambda}}} \,, \quad
\tilde{l}_4 \,=\, \frac{y_2-y_3-1-\sqrt{\tilde{\lambda}}}{y_2-y_3-1+\sqrt{\tilde{\lambda}}} \,. 
\end{align}
In the Feynman integral literature, the arguments $\tilde{l}_i$ of the logarithms in the Fuchsian Pfaffian system~\eqref{eq:fuchsian} are called \textit{symbol letters}, and their ensemble $\{\tilde{l}_i\}$ \textit{symbol alphabet}~\cite{Goncharov:2010jf}. They encode the possible singular points of the solution. Indeed, the letters in \Cref{eq:alphabet} vanish exactly on the singular locus of $I_3^y$ given in \eqref{eq:singI3y}. The system~\eqref{eq:fuchsian} is therefore in a manifestly Fuchsian form, i.e., it has at most simple poles at every singular point.
In~\eqref{eq:alphabet}, the readers with some experience in Feynman integrals may recognize the alphabet of the three-mass triangle Feynman integrals~\cite{Chavez:2012kn}, minus an additional letter equal to~$\sqrt{\tilde{\lambda}}$. The relation between this alphabet and conformal symmetry was already observed in~\cite{CHZ23}, and is a promising hint that this alphabet may be valid at any loop order.
The matrix in~\eqref{eq:By} exhibits the block triangular structure observed in~\cite{Caron-Huot:2014lda} for the differential equations satisfied by finite Feynman integrals. We expect that differential equations corresponding to~\eqref{eq:fuchsian} can be obtained using the method of~\cite{Caron-Huot:2014lda}.

Before we move on to constructing the asymptotic series solutions in the next section, we point out that the form of the first-order Pfaffian system in \eqref{eq:fuchsian} is particularly well-behaved; its solutions can be written down explicitly in terms of logarithms and dilogarithms. The matrix $\tilde{B}(y)$~\eqref{eq:By} is upper block-triangular, which allows us to solve the system~\eqref{eq:fuchsian} iteratively.
We obtain
\begin{align}
\frac{1}{\sqrt{\tilde{\lambda}}} \cdot \tilde{G}(y_2,y_3) \, = \,  \begin{bmatrix} 
 k_1 f_1(x) + k_2 f_2(x) + k_3 f_3(x) + k_4 f_4(x) \\ 
 k_1 \frac{\log\left(\frac{x_2}{x_1}\right)}{\sqrt{\lambda}} + 2 (k_3-k_2) f_4(x) \\ 
 k_1 \frac{\log\left(\frac{x_3}{x_1}\right)}{\sqrt{\lambda}} - 2 k_2 f_4(x) \\ 
 k_1 f_4(x) \\
\end{bmatrix}_{x_1 = 1,\,  x_2 = y_2 ,\, x_3 = y_3} \,,
\end{align}
where $f_1(x), \ldots, f_4(x)$ are the solutions to $I_3$ from~\eqref{eq:f1234}, and $k_1,\ldots, k_4\in \CC$ are arbitrary constants.
Recalling the relation~\eqref{eq:G2F} between $\tilde{F}(y)$ and $\tilde{G}(y)$ with $\tilde{T}(y)$ as in~\eqref{eq:T}, and that the first component of $\tilde{F}(y)$ divided by $y_1$ is a general solution to $I_3$, we have proven that $f_1(x), \ldots, f_4(x)$ are indeed a basis of the space of solutions to~$I_3$. 

\subsection{Asymptotic series solutions}\label{sec:Asymptotic}
\subsubsection{Capturing the weight via an auxiliary variable}
We now compute the asymptotic solutions to the first-order Pfaffian system~\eqref{eq:fuchsian}. Since the latter is in Fuchsian form, we can use Wasow's algorithm~\cite{Wasow} to expand around any singular point without resorting to Gr\"obner bases.
In order to build a bridge to the canonical series solutions,
which are defined with respect to a weight vector~$w\in \RR^n$, we introduce an auxiliary variable~$t$ as follows. Let $f(x)$ be the general solution to the ideal under consideration. We define a new function~$f_w$, which depends both on the original variables $x$ and on $t$ as
\begin{align}\label{eq:auxfun}
f_w\left(t, x\right) \, \coloneqq \, f\left(t^w x\right) \,,
\end{align}
where $t^w x$ is a short-hand notation for $(t^{w_1} x_1,\ldots, t^{w_n} x_n)$.

\begin{remark}
To motivate the construction in \eqref{eq:auxfun}, consider the toy example where the function is a monomial, say $f(x) = x^a$ for some $a\in \NN^n$. The exponent of $t$ in the auxiliary function $f_w(t,x) = t^{w\cdot a} x^a$ gives the $w$-weight of the monomial, namely $w\cdot a$. In this sense, the exponent of the auxiliary variable $t$ captures the notion of weight.
\end{remark}
\noindent We then compute the asymptotic expansion of $f_w(t,x)$ around~$t=0$,
\begin{align} \label{eq:ftildet}
{f}_w(t,x) \,=\, \sum_{k \ge 0} \sum_{m=0}^{m_{\text{max}}} c_{k,m}(x) \, t^k \log(t)^m \,,
\end{align}
where $m_{\text{max}}$ is a
natural number which depends on the specific ideal. By construction, the monomials in $c_{k,m}(x)$ have $w$-weight~$k$. By definition, we have that $f_w|_{t=1} \equiv f.$ Hence, we expect that the asymptotic expansion~\eqref{eq:ftildet} around $t=0$, truncated at $t^k$ and evaluated at $t=1$ equals the canonical series expansion truncated at $w$-weight~$k$.

\begin{example}
Consider the simplest of the  solutions to the ideal $I_3$ from \Cref{eq:f1234},
\begin{align}
f_4(x) \, = \, \frac{1}{\sqrt{\lambda}} \,,
\end{align}
where $\lambda$ is the homogeneous polynomial of degree $2$ defined in \Cref{eq:lambda}.
For the weight vector $w=(-1,0,1)$ in cone $C_1$, we introduce the auxiliary variable $t$ via
\begin{align}
f_{4,w}(t,x) \, = \, f_4\left(\frac{x_1}{t}, \, x_2, \, t \, x_3 \right) \,.
\end{align}
The asymptotic expansion around $t=0$ is given by a Taylor expansion in $t$,
\begin{align} \label{eq:f4exp}
{f}_{4,w}(t,x)  \, = \, \frac{t}{x_1} + t^2 \frac{x_2}{x_1^2} + t^3 \left( \frac{x_2^2}{x_1^3} + \frac{x_3}{x_1^2}\right) + t^4 \left( \frac{x_2^3}{x_1^4} + 4 \frac{x_2 x_3}{x_1^3} \right) + \mathcal{O}\left(t^5 \right) \,.
\end{align}
We see that the monomials $x^a$ appearing as coefficients of $t^k$ have $w$-weight $w \cdot a = k$. Truncating the expansion~\eqref{eq:f4exp} at order $k$ and setting $t=1$ gives the canonical series solution truncated at $w$-weight $k$, e.g.,\
\begin{align} 
f_4(x)  \, = \, \frac{1}{x_1} + \frac{x_2}{x_1^2} + \frac{x_2^2}{x_1^3} + \frac{x_3}{x_1^2}+  \frac{x_2^3}{x_1^4} + \frac{4 x_2 x_3}{x_1^3}  + \cdots \,,
\end{align}
where the dots denote terms of $w$-weight $5$ or higher. Indeed, this coincides with the series $\tilde{f}_1(y_2,y_3)$ in~\eqref{eq:canSeriesC1}, upon substituting $y$ in terms of $x$ in the latter, and dividing it by $x_1$ to obtain a solution of $I_3$ through \Cref{lem:twovarprop}.
\end{example}

The derivatives of $f_w\left(t, x\right)$ with respect to $x$ and $t$ can be obtained from the derivatives of $f(x)$ through the chain rule, so that we can straightforwardly construct a system of PDEs to which $f_w(t,x)$ is a solution, starting from that for~$f(x)$. In the next section, we will use this approach to compute the canonical series solutions to the Fuchsian system derived in \Cref{sec:Fuchsian}, and verify that they reproduce those computed in \Cref{sec:solI3Groebner}.

\subsubsection{Computations for the cone \texorpdfstring{$C_1$}{C\_1}}\label{sec:WasowC1}
We now turn to the asymptotic solution of the Fuchsian system~\eqref{eq:fuchsian}. We choose the weight vector $w=(-1,0,1)$ from the cone $C_1$ of the small Gr\"obner fan. The corresponding weight for the $y$-variables is $w_y = (-1,1,2)$. We thus introduce the auxiliary variable $t$ via
\begin{align} \label{eq:PDEsty}
\tilde{G}_w(t,y_2,y_3) \, = \, \tilde{G}\left(t \, y_2, t^2 y_3 \right) ,
\end{align}
where $\tilde{G}(y_2,y_3)$ is a holomorphic vector-valued function which satisfies the Fuchsian Pfaffian system~\eqref{eq:fuchsian}. The auxiliary function $\tilde{G}_w(t,y)$ satisfies a Pfaffian system in Fuchsian form; for $i=y_2,y_3,t$, 
\begin{align} 
\partial_{i} \bullet \tilde{G}_w(t,y) \, = \, \tilde{B}_{i}(t,y) \cdot \tilde{G}_w(t,y) \, .
\end{align} The matrices $\tilde{B}_{i}(t,y)$ are obtained from the derivatives of $\tilde{G}(y)$ through the chain rule,
\begin{align}
\tilde{B}_i(t,y) \, = \, \pt_i \bullet \tilde{B}\left(t \, y_2, t^{2} y_3 \right) \,.
\end{align}
In order to compute the asymptotic expansion of the solutions around $t=0$, we need to study the behavior of the matrices $\tilde{B}_{i}(t,y)$ around $t=0$. The Laurent expansion of $\tilde{B}_t(t,y)$~is
\begin{align} \label{eq:tildeBtexp}
\tilde{B}_t(t,y) \ =\ \frac{\tilde{B}_t^{\text{res}}}{t} \,+\, \sum_{k\ge 0} t^k \, \tilde{B}_t^{(k)}(y)  \,.
\end{align}
Since the system is in Fuchsian form, there can be no pole of higher order.
The residue at \mbox{$t=0$}, i.e., the constant matrix $\tilde{B}_t^{\text{res}}$, has the 
eigenvalue~$0$ only. The matrices $\tilde{B}_{y_2}(t,y)$ and $\tilde{B}_{y_3}(t,y)$ are instead non-singular at~$t=0$.

The first step of the algorithm in~\cite{Wasow} is to perform a gauge transformation  $\tilde{G}_w \rightarrow \tilde{G}_w'$,
\begin{align}
\tilde{G}_w(t,y) \,=\, \tilde{U}(t,y) \cdot \tilde{G}_w'(t,y)  \,.
\end{align}
In the new basis, our vector-valued function satisfies the system
\begin{align}
\partial_i \bullet \tilde{G}_w'(t,y) \,=\, \tilde{B}'_i(t,y) \cdot \tilde{G}_w'(t,y) \,,
\end{align}
for $i=y_2,y_3,t$, with 
\begin{align} \label{eq:tildeBprime}
\tilde{B}'_i(t,y) \, = \, \tilde{U}^{-1}(t,y) \cdot \tilde{B}_i(t,y) \cdot \tilde{U}(t,y) - \tilde{U}^{-1}(t,y) \cdot \partial_i\bullet \tilde{U}(t,y)  \,.
\end{align}
The key idea is to choose the gauge transformation such that the DEs for $\tilde{G}_w'(t,y)$ are simpler than those for $\tilde{G}_w(t,y)$ in view of the asymptotic solution around $t=0$. In particular, we wish to simplify $\tilde{B}'_t(t,y)$ so that it has a simple pole only, i.e., is of the form
\begin{align} \label{eq:desiredform}
\tilde{B}'_t(t,y) \, = \, \frac{\tilde{B}_t^{\text{res}}}{t} \, .
\end{align}
We can construct the transformation matrix $\tilde{U}(t,y)$ which achieves~\eqref{eq:desiredform} as a series in $t$,\footnote{The construction assumes that the matrix $\tilde{B}_t^{\text{res}}$ has no eigenvalues that differ from each other by positive integers \cite[Theorem 5.1]{Wasow}. If this condition is not satisfied, a less specific result holds \cite[Theorem 5.6]{Wasow}.}
\begin{align} \label{eq:Uexp}
\tilde{U}(t,y) \, = \, \sum_{k\ge 0} t^k \, \tilde{U}_k(y) \,.
\end{align}
We plug this expansion into \Cref{eq:tildeBprime} and impose that the resulting $\tilde{B}'_t(t,y) $ satisfies~\eqref{eq:desiredform}. The first term in the expansion, $\tilde{U}_0(y)$, must commute with $\tilde{B}_t^{\text{res}}$.
The simplest way to achieve this is to choose $\tilde{U}_0(y) = \mathbb{I}_4$ to be the $4 \times 4$ identity matrix.
The terms of expansion~\eqref{eq:Uexp} of higher order are determined recursively in terms of the lower ones through the solution of a linear system of equations,
\begin{align}
\tilde{U}_k(y) \cdot \tilde{B}_t^{\text{res}} - \tilde{B}_t^{\text{res}} \cdot \tilde{U}_k(y) + k \, \tilde{U}_k(y) \  = \ \sum_{r=0}^{k-1} \tilde{B}_t^{(k-r-1)}(y) \cdot \tilde{U}_r(y)  \,.
\end{align}
As an example, we spell out the first few terms of the resulting transformation matrix:
\begin{align}
\tilde{U}(t,y) \, = \, \mathbb{I}_4 + t \begin{bmatrix}
0 & - y_2 & 0 & y_2 \\
0 & 0 & 0 & 0 \\
0 & 0 & 0 & 0 \\
0 & 0 & 0 & 0 \\
\end{bmatrix} + t^2 \begin{bmatrix}
0 & -\frac{1}{2} y_2^2 & - y_3 & \frac{1}{4} y_2^2 + y_3 \\
0 & 0 & 0 & 0 \\
0 & 0 & 0 & 0 \\
0 & 0 & 0 & 0 \\
\end{bmatrix} + \mathcal{O}\left(t^3\right) \,.
\end{align}
As expected, each order in $t$ involves only monomials whose $w$-weight matches the power of~$t$.
As a result, $\tilde{G}_w'(t,y)$ satisfies the following simplified system of PDEs,
\begin{align} \label{eq:GprimePDEs}
\begin{cases}
\partial_{y_2} \bullet \tilde{G}_w'(t,y) &\, = \ \tilde{B}'_{y_2}\left(t,y\right) \cdot \tilde{G}_w'(t,y) \,, \\
\partial_{y_3} \bullet \tilde{G}_w'(t,y) &\, = \ \tilde{B}'_{y_3}\left(t,y\right) \cdot \tilde{G}_w'(t,y) \,, \\
\partial_t \bullet \tilde{G}_w'(t,y) &\, = \ \dfrac{\tilde{B}_t^{\text{res}}}{t} \cdot \tilde{G}_w'(t,y) \,, \\
\end{cases}
\end{align}
with $\tilde{B}'_{y_2}\left(t,y\right)$ and $\tilde{B}'_{y_3}\left(t,y\right)$ determined by~\eqref{eq:tildeBprime}, and non-singular at $t=0$.
We can solve the simplified system~\eqref{eq:GprimePDEs} using the path-ordered exponential formalism (see e.g.~\cite{Brown:2013qva} and the references therein). We integrate this system along a path in the $(t,y)$-space of the form
\begin{align}
\left( 0 \,, y^{(0)} \right) \longrightarrow \left( 0 \,, y \right) \longrightarrow \left( t, y \right) ,
\end{align} 
for some arbitrary values $y^{(0)}$ of $y$ which do not belong to the singular locus of the ideal. In other words, we pick $(t=0, y=y^{(0)})$ as boundary point,\footnote{Since $t=0$ is a regular singular point of the DEs, this is a \textit{tangential} boundary point. We refer the interested readers to \cite[Section~2.5]{Brown:2013qva}.} restore the dependence on $y$ by integrating along the first piece of the path, and then on $t$ by integrating along the second.
The result is
\begin{align} \label{eq:GprimeSol}
\tilde{G}_w'(t,y) \, = \, \text{e}^{\tilde{B}_t^{\text{res}} \log(t)} \cdot \tilde{h}(y) \,,
\end{align} 
where 
\begin{align}
\text{e}^{\tilde{B}_t^{\text{res}} \log(t)} \, = \, \begin{bmatrix}
1 & -\log(t) & -\frac{1}{2} \log(t) & -\log^2(t) \\
0 & 1 & 0 & \log(t) \\
0 & 0 & 1 & 2 \log(t) \\
0 & 0 & 0 & 1\\
\end{bmatrix} \,,
\end{align}
and $\tilde{h}(y)$ is the vector-valued holomorphic function which solves the system
\begin{align}
\begin{cases} 
\partial_{y_2} \bullet \tilde{h}(y) \, = \, \tilde{B}'_{y_2} ( t=0,y) \cdot \tilde{h}(y) \,, \\
\partial_{y_3} \bullet \tilde{h}(y) \, = \, \tilde{B}'_{y_3} ( t=0,y) \cdot \tilde{h}(y) \,. \\
\end{cases}
\end{align}
Explicitly, we have
\begin{align} \label{eq:h}
\tilde{h}(y) \, = \, \begin{bmatrix}
 1 &  -\frac{1}{2} \log(y_3) & -\frac{1}{2}\log(y_2) & -\frac{1}{2} \log(y_2) \log(y_3) \\
 0 & 1 & 0 & \log(y_2) \\
 0 & 0 & 1 & \log(y_3) \\
 0 & 0 & 0 & 1 \\
\end{bmatrix}  \cdot b  \,,
\end{align}
with $b = (b_1,b_2,b_3,b_4)^{\top}\in \mathbb{C}^4$.

We now have all the ingredients to compute the asymptotic expansion around $t=0$ of the solution of the first-order Pfaffian and Fuchsian system of DEs~\eqref{eq:fuchsian}. It is given by
\begin{align} \label{eq:Gexp}
\tilde{G}_w(t,y) \, = \, \tilde{U}(t,y) \cdot \text{e}^{\tilde{B}_t^{\text{res}} \log(t)} \cdot \tilde{h}(y) \,.
\end{align}
The gauge transformation matrix, $\tilde{U}(t,y)$, is given by a Taylor series around $t=0$ starting with the identity matrix, and can be computed algorithmically up to any order in $t$. The matrix exponential, $\text{e}^{\tilde{B}_t^{\text{res}} \log(t)}$, is irrelevant for our purposes, as we will eventually set $t=1$, and $ \text{e}^{\tilde{B}_t^{\text{res}} \log(1)} = \mathbb{I}_4$. The matrix-valued function $\tilde{h}(y)$ given in \Cref{eq:h} is therefore the only source of powers of $\log(y)$ in this approach. 

\smallskip

We can now write down the asymptotic expansions of the solution to the ideal $I_3^y$. We recall that a solution to the latter, $\tilde{f}(y)$, is by construction given by the first component of the vector-valued function $\tilde{F}(y)$ (see \eqref{eq:Fdef}). The latter is related to the vector-valued function $\tilde{G}(y)$ through the gauge transformation~\eqref{eq:G2F} with transformation matrix $\tilde{T}(y)$. Putting everything together, we have
\begin{align}
\tilde{f}(y) \,=\, \left[ \tilde{T}\left(y \right) \cdot \tilde{G}(y) \right]_1 \,,
\end{align} 
where the subscript $1$ denotes the first component of the vector. The series expansion of $\tilde{G}(y)$ truncated at $w$-weight $k$ is given by the asymptotic expansion of $\tilde{G}_w(t,y)$ around $t=0$ up to order $k$, given by~\eqref{eq:Gexp}, setting $t=1$. Similarly, the series expansion of the transformation matrix is obtained by the Taylor expansion of $\tilde{T}_w(t,y) = \tilde{T}(t^w y)$ around $t=0$, truncated at~$t^k$ and evaluated at $t=1$.
Then, the four linearly independent series solutions to $I^y_3$ truncated at $w$-weight~$3$~are
{
\begin{align}\label{eq:wasowexpansions}\begin{split}
\tilde{\mathfrak{h}}_1(y) &\,=\,  \left[1  \right] + \left[y_2\right] + \left[ y_2^2 
 + y_3\right] + \cdots  \,, \\
\tilde{\mathfrak{h}}_2(y) &\,=\,  \left[  \log(y_2) \right] + \left[y_2 \log(y_2)  
 \right] + \left[ 2 y_3 + (y_2^2 + y_3)  \log(y_2) \right] + \cdots  \,, \\
\tilde{\mathfrak{h}}_3(y) &\,=\, 
 \left[\log(y_3)\right] + \left[ 2 y_2 + y_2 \log(y_3) \right] + 
 \left[ 3 y_2^2 + (y_2^2 + y_3) \log(y_3) \right] +\cdots  \,, \\
\tilde{\mathfrak{h}}_4(y) &\,=\,    
 \left[\log(y_2) \log(y_3) \right] + \left[ -2 y_2 + 2 y_2 \log(y_2) + y_2 \log(y_2) \log(y_3) \right]  \\ 
& \quad \ + \left[ -\frac{5}{2} y_2^2 - 2 y_3 + 3 y_2^2 \log(y_2) + 2 y_3 \log(y_3) + (y_2^2 + y_3) \log(y_2) \log(y_3) \right]\\ 
& \quad \ + \cdots \,.
\end{split} 
\end{align}}

In the above functions, the square brackets highlight terms of different weight, and the dots denote terms of $w$-weight $4$ or higher. From these, we see that the starting monomials coincide with the solutions to the indicial ideal in~\eqref{eq:indsol1}. The series expansions in~\eqref{eq:wasowexpansions} match those computed by Gr\"obner bases in~\eqref{eq:canSeriesC1} (with $\tilde{\mathfrak{h}}_i(y) = \tilde{f}_i(y)$),
and are related to the known solutions $f_i$ from~\eqref{eq:f1234} via
\begin{align}
\begin{split}
&f_1(1,y_2,y_3) \, = \, - \tilde{\mathfrak{h}}_4(y) - \frac{\pi^2}{3} \tilde{\mathfrak{h}}_1(y) \,, \\
&f_3(1,y_2,y_3)  \, = \, \tilde{\mathfrak{h}}_3(y)  \,, 
\\ 
&f_2(1,y_2,y_3)  \, = \, -\tilde{\mathfrak{h}}_2(y)-\tilde{\mathfrak{h}}_3(y) \,,  \\ 
& f_4(1,y_2,y_3) \, = \, -\tilde{\mathfrak{h}}_1(y) \,. \\
\end{split}
\end{align}
We thus find agreement among the known solutions, the canonical series solutions computed via the SST algorithm, and those computed by solving the Pfaffian system asymptotically. The computations for the remaining cones of the Gr\"{o}bner fan follow the same strategy. The results can be obtained by relabelling the indices as discussed at the end of \Cref{sec:solI3Groebner}.

\section{Conclusion and outlook}\label{sec:outlook}
In this paper, we studied algorithms for obtaining solutions of linear partial differential equations relevant for particle physics. In principle, this applies to all Feynman integrals, as the latter are known to satisfy systems of differential equations (see e.g.~\cite{Muller-Stach:2012tgj}), as well as to further special functions appearing as their solutions, such as multiple polylogarithms. 
We implemented a method, originally proposed by Saito, Sturmfels, and Takayama~\cite{SST00}, to compute canonical Nilsson series expansions. The SST algorithm had been in use already~\cite{Cruz19,TH22} for GKZ systems. To our knowledge, our paper is the first one to implement the algorithmic ideas of SST without requiring hypergeometricity. Using the particular example of conformal differential equations satisfied by a one-loop triangle Feynman integral with general propagator powers, we implemented the SST method and compared it to a method of Wasow~\cite{Wasow} that is more commonly used in the context of differential equations for Feynman integrals~\cite{Henn:2014qga,Lee:2014ioa}.

Using the example of the triangle Feynman integral, we showed how methods \cite{SST00,SatStu19} from the theory of $D$-modules can be used to determine key features of the differential equations, as well as canonical series expansions of the solutions. In particular, for holonomic systems, the holonomic rank corresponds to the number of master integrals in the physics literature. The singular locus determines the set of (possible) singularities of the corresponding Feynman integrals. The solutions to the indicial ideal are the starting terms of the series. Following SST, we use Gr\"obner basis methods to obtain canonical series expansions. Leveraging the techniques developed for computing Feynman diagrams, we present alternative methods which allow us to extract this precious information from the PDEs without relying on the computation of Gr\"obner bases.

It is interesting to ask about the scope of the method. Proxies for the complexity of the differential equations are: the number of variables; the order of the differential operators; the number of differential equations. A potential bottleneck in the SST approach could be the reliance on Gr\"obner bases. 
However, compared to other situations where Gr\"obner methods are used (e.g.\ for analyzing equations of high polynomial degree), one would expect the polynomials appearing in the PDEs of physical interest to have moderate degree.
 (As a proof of principle, we showed in Appendix C an application to a four-loop ladder integral.)
Moreover, our comparison with Wasow's method shows that, in principle, these methods can be replaced by linear algebra operations, which could potentially scale better.
Then again, Wasow's method relies on the Pfaffian system being in Fuchsian form at the regular singular point around which we wish to expand. While several strategies to achieve this have been developed especially in the context of Feynman integrals, they have limitations.
This makes it even more interesting to have two complementary approaches.

There are various interesting practical and conceptual questions for follow-up work. On the practical side, the physics literature provides many important classes of Feynman integrals that are relevant for the phenomenology of elementary particles, in the context of gravitational wave physics, or in cosmology, for example. A prerequisite of our method is the knowledge of the relevant differential equations. Obtaining the latter is an interesting active area of research in itself, see e.g.~\cite{LV22} and the references therein. It would be interesting to study the scope of the SST algorithm for the Picard--Fuchs equations obtained in~\cite{LV22}.

On the conceptual side, let us mention the following questions. Firstly, when dealing with equations in multiple variables, one may ask how one can restrict the differential equations to lower-dimensional submanifolds. Physical examples include on-shell limits, or restrictions to submanifolds corresponding to (spurious) singularities. Secondly, there are important situations where only a subset of the relevant differential equations is available (e.g.\ because they are easier to obtain than the complete equations, as e.g.\ in~\cite{Drummond:2010cz}). In other words, in these situations, the holonomic rank of the $D$-ideal is not finite, and hence there is functional freedom in the solutions. This happens for example for the differential equations that follow from conformal or Yangian symmetry beyond the three-particle case. It is then interesting to study how holonomic techniques can be used to obtain useful information, such as which additional constraints may be added to make the system holonomic.
Finally, the triangle Feynman integral considered here has a finite value, but Feynman integrals typically exhibit divergences in the infrared and ultraviolet regions of the loop integration. 
 In dimensional regularization, these divergences are regulated by analytically continuing to generic $d\in \mathbb{C}$ spacetime dimensions. The divergences are then manifested as poles in the Laurent expansion around $d=d_0$ for some integer $d_0$ (typically $d_0=4$). The coefficients of this Laurent expansion are solutions to regular holonomic $D$-ideals in the kinematic variables and, as such, can in principle be treated using the $D$-module techniques discussed here. It would therefore be important to find a way to set up the expansion around $d=d_0$ within the $D$-module approach. 

\medskip

Beyond these specific research questions, we believe that the exchange of methods promoted in this work will be fruitful for both the communities of mathematicians and theoretical physicists. Given the ubiquity of PDEs, we envision these methods will be beneficial for many further applications. 

\subsection*{Data Availability Statement.} 

Our implementations are made available via the MathRepo~\cite{mathrepo}---a repository website hosted by MPI~MiS---at \url{https://mathrepo.mis.mpg.de/DModulesFeynman}. It follows the FAIR data principles of the mathematical research data initiative \href{https://www.mardi4nfdi.de/about/mission}{MaRDI}~\cite{Mardi}, 
which aim to improve {\bf f}indability, {\bf a}ccessibility, {\bf i}nteroperability, and {\bf r}euse of digital~assets. 

\subsection*{Acknowledgments.}
We thank Bernd Sturmfels and Clément Dupont for several insightful discussions, MPI-MiS Leipzig for their support and hospitality, and the referees for their helpful comments. \mbox{A.-L.~S.} was partially supported by the Wallenberg AI, Autonomous Systems and Software Program (WASP) funded by the Knut and Alice Wallenberg Foundation. This project received funding from the European Union’s Horizon 2020 research and innovation programs {\em Novel structures in scattering amplitudes} (grant agreement No.~725110) and {\em High precision multi-jet dynamics at the LHC} (grant agreement No.~772099), and from the European Union’s Horizon Europe Research and Innovation Programme under the Marie Skłodowska-Curie grant agreement No.~101105486. This research was supported by the Munich Institute for Astro-, Particle and BioPhysics (MIAPbP), which is funded by the Deutsche Forschungsgemeinschaft (DFG, German Research Foundation) under Germany’s Excellence Strategy -- EXC-2094 -- 390783311.

\appendix 

\section{Conformal symmetry}\label{appendix:conformal}
Conformal symmetry underlies many of the quantum field theories which are important to our understanding of fundamental physics: massless quantum electrodynamics, quantum chromodynamics, Yang--Mills theory, scalar $\phi^4$ theory, and more, all have conformal symmetry at the classical level.\footnote{Understanding the interplay between quantum corrections and conformal symmetry is object of active research, see see for instance \cite{Braun:2003rp,CHZ23} and the references therein.}
Despite its ubiquity, little is known about its implications for on-shell scattering amplitudes. The latter are functions of the particles' momenta which play a central role in describing the interactions of fundamental particles. One of the obstacles which has hindered the study of this topic is that the constraints imposed by conformal symmetry have the form of second-order differential operators in momentum space.  Determining the conformally-invariant functions relevant for the scattering of particles thus requires to solve systems of linear second-order PDEs. Given the large number of variables which describe the particles' scattering, this problem is challenging from a technical point of view. As such, it is an optimal testing ground for $D$-module techniques. 

The conformal symmetry group is an extension of the Poincar\'e group $\operatorname{SO}(1, 3) \ltimes \RR$, which is the semi-direct product of the Lorentz group and the Abelian group of translations. The Poincar\'e group plays a central role in fundamental physics, as it is the symmetry group of Einstein's theory of special relativity, and hence of all theoretical models which stem from it. The conformal group is obtained by extending the Poincar\'e group by two classes of transformations: dilatations and conformal boosts (alias ``special conformal transformations''). We refer the interested readers to~\cite{DiFrancesco:1997nk} for a thorough discussion of conformal symmetry in field theory, and present here what is necessary to introduce the system of PDEs addressed in this work.
We begin by introducing some notation. We denote by
\begin{align}
 z  \,=\,\left( z^0 , z^1, \ldots, z^{d-1} \right)^{\top} \,,
\end{align}
the vector of $d$-dimensional spacetime coordinates, where $z^0$ is the time coordinate, and $z^i$ for $i=1,\ldots,d-1$ are the spatial coordinates. We define the scalar product of two coordinate vectors $z_1$ and $z_2$ as
\begin{align}
z_1 \cdot z_2 \, \coloneqq \, z_1^{\top} \cdot g \cdot z_2 \,,
\end{align}
where $g = \text{diag}(1,-1,\ldots,-1)$ is called {\em metric tensor}, and define $|z|^2 \coloneqq z\cdot z$.

The conformal transformations are given by the composition of the coordinate transformations shown in Table~\ref{tab:conformal}.\footnote{A conformal boost can also be viewed as the composition of an inversion ($z \to z/|z|^2$), followed by a translation, and another inversion.} They form a group with respect to composition.
A function is {\em conformally-invariant} if it is annihilated by the generators of the conformal group. The latter are first-order differential operators in the spacetime coordinates which obey the commutation relations of the Lie algebra associated with the conformal group. The generator of the group of dilatations for $n$ points of coordinates $z_1, \ldots,z_n$ is for instance given by
\begin{align}
\mathfrak{D}_n \,=\, - \mathrm{i} \sum_{k=1}^n \left( z_k \cdot \partial_{z_k} + c_k \right) \,,
\end{align}
where the $c_k \in \mathbb{R}$ are parameters called {\em conformal weights}, and $\partial_{z_k}$ denotes the vector of partial derivatives,
\begin{align} \partial_{z_k} = \left( \frac{\partial}{\partial z_k^{\mu}}\right)_{\mu \,=\, 0,\ldots,d-1}^{\top} \,.
\end{align}

\bgroup
\def\arraystretch{1.5}
\begin{center}\begin{table}[t]
\begin{tabular}{|ll|}
\hline
Translations & $z \longrightarrow z + \epsilon$, \  $\epsilon\in\mathbb{R}^d$ \\  
(Proper) Lorentz transformations &
$z \longrightarrow \Lambda \cdot z$, \ $\Lambda \in \text{SO}(1,d-1)$ \\
Dilatations & $z \longrightarrow \text{e}^{\omega} \, z$, \ $\omega \in \mathbb{R}$ \\
Conformal boosts & $z \longrightarrow  \dfrac{z - |z|^2\, \epsilon}{1 - 2 \, z \cdot \epsilon + |z|^2 |\epsilon|^2 }$, \ $\epsilon \in \mathbb{R}^d$ \\
\hline
\end{tabular}
\caption{Conformal transformations of the spacetime coordinate $z$.}
\label{tab:conformal}
\end{table}\end{center}
\egroup
Scattering amplitudes are functions of the particles' momenta. Like the spacetime coordinates $z_k$, the momenta $p_k$ are $d$-dimensional vectors. Their zeroth components, $p_k^0$, give the particles' energy, while the other components, $p_k^{\mu}$ with $\mu=1,\ldots,d-1$, are related to their velocity. Coordinate and momentum space are related by a Fourier transform.

The momentum-space realization of the generators of the Poincar\'e group are given by first-order operators also in the momenta $p_k$. The constraints they impose are particularly simple. The invariance under translations implies the conservation of total momentum, namely that $p_1+\cdots+p_n = 0$. We can use this constraint to eliminate one of the momenta, say $p_n$. The invariance under Lorentz transformations instead implies that the invariant functions depend on the momenta only through the scalar products $p_k \cdot p_{\ell}$, called {\em Mandelstam invariants} in the physics literature. Therefore, a generic Poincar\'e invariant function of $n$ momenta $p_1,\ldots,p_n$ depends only on $n(n-1)/2$ variables, which can be chosen as
\begin{align}
 \mathfrak{s}_n \, \coloneqq \, \left\{ p_k \cdot p_{\ell} \,|\, 1 \le k \le \ell \le n-1 \right\} \,.
 \end{align}
From here onwards, we will thus consider functions of such variables, and focus on dilatations and conformal boosts. The momentum-space realization of the generator
of dilatations is
\begin{align}\label{eq:Dmomentum}
    \widehat{\mathfrak{D}}_n \,=\, \mathrm{i} \sum_{k=1}^n \left( p_k \cdot   \frac{\partial}{\partial p_{k}} + d - c_k \right) \,,
\end{align}
and those of conformal boosts is\footnote{We focus on Lorentz-scalar quantities. The form of the conformal boost generator for functions which transform in a non-trivial representation of the Lorentz group has an additional first-order term (see e.g.~\cite{DiFrancesco:1997nk}).}
\begin{align} \label{eq:Kmu}
  \widehat{\mathfrak{K}}_n \,=\, \sum_{k=1}^n \left[ -p_k \, (\partial_{p_{k}} \cdot \partial_{p_{k}}) + 2 \, (p_k \cdot \partial_{p_k}) \, \partial_{p_k} + 2 \, (d-c_k) \, \partial_{p_{k}} \right] \, .
\end{align}
 While the generator $\widehat{\mathfrak{D}}_n$ of the group of dilatations is a first-order operator also in momentum space, the generator $\widehat{\mathfrak{K}}_n$ of the group of conformal boosts is now given by a second-order operator. Determining the conformally-invariant functions for a given $n$ therefore amounts to solving the system of linear second-order PDEs
\begin{align}
\widehat{\mathfrak{D}}_n \bullet f\left(\mathfrak{s}_n\right) \,=\, 0 \qquad \text{and} \qquad \widehat{\mathfrak{K}}_n \bullet f\left(\mathfrak{s}_n\right) \,=\, 0 \,.
\end{align}
We can rewrite all generators as differential operators with respect to the set of variables $\mathfrak{s}_n$ using the chain rule for any $n$. 

For this first study, we focused on the case of $n=3$ particles with momenta $p_i$. We choose $\mathfrak{s}_3 = \{x_1,x_2,x_3\}$ with $x_i = |p_i|^2$ as the minimal set of variables to implement Poincar\'e symmetry in momentum space, and use the chain rule to express the generators of dilatations and conformal boosts as differential operators in them. The operators $P_1$ and $P_2$ in \Cref{I3system} stem from $\hat{\mathfrak{K}}_3$, while $P_3$ stems from $\hat{\mathfrak{D}}_3$.
The space of conformally-invariant functions is easier to determine in the coordinate space; there, the conformal generators are of order one. It is long known (see e.g.~\cite{Polyakov:1970xd}) that one solution is given by
\begin{align} \label{eq:coordinate_solution}
 \frac{1}{\left(|z_1-z_2|^2\right)^{\frac{c_1+c_2-c_3}{2}} \left(|z_2-z_3|^2\right)^{\frac{c_2+c_3-c_1}{2}} \left(|z_3-z_1|^2\right)^{\frac{c_3+c_1-c_2}{2}}} \,.
\end{align}
The Fourier transform of this function to momentum space can be expressed as the one-loop triangle integral $J^{\text{triangle}}_{d;\nu_1,\nu_2,\nu_3}$ from \Cref{eq:triangle_mom}, with the exponents $\nu_i$ being related to the conformal weights $c_i$ via \Cref{eq:c2nu}. This constitutes the link between the triangle Feynman integrals and conformal symmetry, which is central to \Cref{sec:confdiffeq}.

\section{Integrable connections and Pfaffian systems}\label{appendix:IntegrPfaf}
The presentation of this section follows mainly \cite{HTT08}.
Let $X$ be a smooth algebraic variety. Denote by $\mathcal{O}_X$ the sheaf of regular functions on $X$ and by $\Theta_X$ the tangent sheaf on $X$. The sheaf of linear differential operators on $X$, denoted $\mathcal{D}_X$, is the sheaf of subalgebras of $\mathcal{E}\text{nd}_{\mathbb{C}_X}(\mathcal{O}_X)$ generated by $\mathcal{O}_X$ and $\Theta_X$. Regular functions are differential operators of order~$0$.

\begin{example}\label{ex:Daffine}
Let $X$ be the affine $n$-space $\mathbb{A}_{\CC}^n$. The global sections of its sheaf of regular functions are polynomials, i.e., $\mathcal{O}_X(X)=\CC[x_1,\ldots,x_n]$, and $\mathcal{D}_X(X)=D_n$. Any quasi-coherent $\mathcal{D}_{\mathbb{A}_{\mathbb{C}}^n}$-module $\mathcal{M}$ is generated by its global sections, i.e., by elements of the \mbox{$D_n$-module} $\mathcal{M}(\mathbb{A}_{\mathbb{C}^n})$. In fact, by \cite[Proposition 1.4.4]{HTT08}, quasi-coherent $\mathcal{D}_{\mathbb{A}_{\mathbb{C}}^n}$-modules are in one-to-one correspondence to $D_n$-modules.
\end{example}

Let $\mathcal{M}$ be a vector bundle, i.e., a locally free sheaf of $\mathcal{O}_X$-modules on $X$. In particular, $\mathcal{M}$ is quasi-coherent. 
Giving $\mathcal{M}$ the structure of a $\mathcal{D}_X$-module is equivalent to equipping $\mathcal{M}$ with an {\em integrable connection}, i.e., a $\mathbb{C}_X$-linear morphism
\begin{align}\label{eq:integrconn}
\nabla \colon \, \Theta_X \longrightarrow \mathcal{E}\text{nd}_{\mathbb{C}_X}(\mathcal{M}), \qquad \theta \mapsto \nabla_{\theta}
\end{align}
such that 
\begin{enumerate}[(1)]
    \item $\nabla_{f\theta}(m) = f \nabla_{\theta}(m)$, 
    \item $\nabla_{\theta}(fm)=\theta(f)m+f\nabla_{\theta}(m)$ (Leibniz' rule),
    \item $\nabla_{[\theta_1,\theta_2]}(m) = [\nabla_{\theta_1},\nabla_{\theta_2}](m)$ (integrability),
\end{enumerate}
for all sections $f\in \mathcal{O}_X,$ $m\in \mathcal{M},$ and $\theta,\theta_1,\theta_2 \in \Theta_X$.
\begin{example}[\Cref{ex:Daffine} cont'd]
Let $M$ be the $D_n$-module corresponding to $\mathcal{M}$. Then, the action of $\partial_i$ on an element $m\in M=\mathcal{M}(X)$ is given by $\nabla_{\partial_i}(m),$ i.e., $\partial_i \bullet m = \nabla_{\partial_i}(m).$
\end{example}
\begin{remark}
By the tensor-hom adjunction, an integrable connection can equally be written as
\begin{align}\label{eq:nablaOmega}
\nabla \colon \ \mathcal{M} \longrightarrow \mathcal{M}\otimes \Omega_X^1 \, ,
\end{align}
where $\Omega_X^1$ denotes the cotangent sheaf of~$X$. 
\end{remark}

After choosing a basis of the vector bundle~$\mathcal{M}$, the action of $\partial_i$ is given by a matrix~$A_i$. The flat sections of $\nabla$ in the form \eqref{eq:nablaOmega}, i.e., those local sections $m\in \mathcal{M}$ for which $\nabla(m)=0$, are the solutions to the systems encoded by the~$\partial_i-A_i$'s. In order to compute solutions, one has to consider the analytified connection, since one is typically interested in {analytic} solutions as well. Otherwise, one would ignore solutions involving $\exp$, $\ln$, and so on: they are not regular functions on affine space $\mathbb{A}_{\mathbb{C}^n}$, but they are (multi-valued) analytic functions on (subsets of) $\CC^n$.
If one changes the basis of the vector bundle $\mathcal{M}$ by an invertible $m\times m$ matrix~$U$, the matrices in the new basis become
\begin{align}\label{eq:gauge}
U \cdot A_i \cdot U^{-1} \,+\, \frac{\partial U}{\partial x_i} \cdot U^{-1} \, .
\end{align}
This process is called {\em gauge transformation} and is typically carried out at the stalks of the connection at its singular points.

\section{Reducing the four-loop ladder to solving one ODE}\label{sec:fourloop}

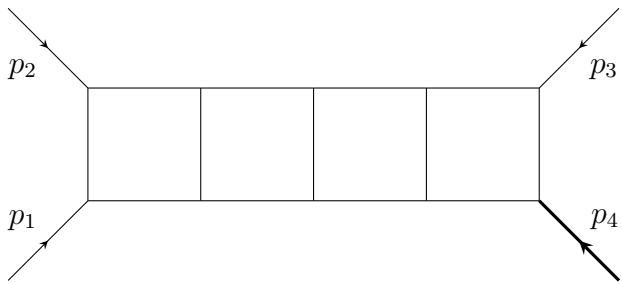
\begin{figure}[t!]
\begin{tikzpicture}[decoration={
    markings,
    mark=at position 0.5 with {\arrow{stealth}}}
    ] 
\coordinate (v1) at (0,0);
\coordinate (v2) at (0,1.5);
\coordinate (v3) at (1.5,1.5);
\coordinate (v4) at (3,1.5);
\coordinate (v5) at (4.5,1.5);
\coordinate (v6) at (6,1.5);
\coordinate (v7) at (6,0);
\coordinate (v8) at (4.5,0);
\coordinate (v9) at (3,0);
\coordinate (v10) at (1.5,0);
\coordinate (p1) at (-1.0607, -1.0607);
\coordinate (p2) at (-1.0607, 1.0607+1.5);
\coordinate (p3) at (6+1.0607, 1.0607+1.5);
\coordinate (p4) at (6+1.0607, -1.0607);
\draw (v1) -- (v2);
\draw (v2) -- (v3);
\draw (v3) -- (v4);
\draw (v4) -- (v5);
\draw (v5) -- (v6);
\draw (v6) -- (v7);
\draw (v7) -- (v8);
\draw (v8) -- (v9);
\draw (v9) -- (v10);
\draw (v10) -- (v1);
\draw (v10) -- (v3);
\draw (v9) -- (v4);
\draw (v8) -- (v5);
\draw[postaction={decorate}] (p1) -- (v1) node[above left,midway,anchor=south east] {$p_1$};
\draw[postaction={decorate}] (p2) -- (v2) node[below left,midway,anchor=north east] {$p_2$};
\draw[postaction={decorate}] (p3) -- (v6) node[below right,midway,anchor=north west] {$p_3$};
\draw[postaction={decorate}, very thick] (p4) -- (v7) node[above right,midway,anchor=south west] {$p_4$};
\end{tikzpicture}
\caption{The Feynman graph 
representing the four-loop ``ladder'' Feynman integral defined in \Cref{eq:ladder4L}. Due to momentum conservation, $p_1+p_2+p_3+p_4=0$.
The momenta corresponding to thin (thick) external legs are on-shell (off-shell), i.e., we have that $|p_i|^2 = 0$ for $i=1,2,3$, and $|p_4|^2 \neq 0$.}
\label{fig:ladder4L}
\end{figure}

In this appendix, we apply the $D$-module methods discussed in the main text to a more complicated Feynman integral:
the four-loop ``ladder'' integral with one off-shell leg,
{\small 
\begin{align}\begin{split}\label{eq:ladder4L}
J^{\text{ladder}}_{d} \,=\, \int
& \left( \prod_{j=1}^4 \frac{\mathrm{d}^d k_j}{\mathrm{i} \pi^{\frac{d}{2}}} \right)
\frac{(|p_1+p_2|^2)^{2 (6-d)} |p_2+p_3|^2}{
|k_1|^2 \, |k_1+p_1|^2 \, |k_1+p_1+p_2|^2 \, |k_2|^2  \, |k_2-k_1|^2 \, |k_2+p_1+p_2|^2 } \, \times \\
& \ \frac{1}{|k_3|^2 \, |k_3-k_2|^2 \, |k_3+p_1+p_2|^2 \, |k_4|^2 \, |k_4-k_3|^2 \, |k_4+p_1+p_2|^2 \, |k_4-p_4|^2}\,,
\end{split}
\end{align}}where the integration domain is the Minkowski spacetime (see \Cref{sec:confdiffeq}).
The corresponding Feynman graph is shown in \Cref{fig:ladder4L}. 
The normalization factor of $(|p_1+p_2|^2)^{2 (6-d)} |p_2+p_3|^2$ is chosen to make the integral dimensionless and simplify the  results.
We take the momentum $p_4$ to be off-shell, i.e.\ $|p_4|^2 \neq 0$.
There are three independent Mandelstam invariants (see \Cref{appendix:conformal}):  $|p_1+p_2|^2$, $|p_2+p_3|^2$,~$|p_4|^2$. 
We combine them in the following variables:
\begin{align} \label{eq:ydef}
    y_1 \,=\, |p_4|^2 \,, \qquad y_2 \,=\, \frac{|p_1+p_2|^2-|p_4|^2}{|p_4|^2} \,, \qquad y_3 \,=\, \frac{|p_4|^2-|p_1+p_2|^2-|p_2+p_3|^2}{|p_4|^2} \, .
\end{align}
Thanks to the chosen normalization, the ladder integral depends only on $y_2$ and $y_3$.

The state of the art for Feynman integrals with these kinematics is three loops, cf.~\cite{DiVita:2014pza,Canko:2021xmn,Henn:2023vbd}. 
Any information about its analytic structure would therefore be of great interest.
For this reason, we focus on its {\em maximal cut} in $d=4$ dimensions.
The maximal cut---which amounts to replacing all propagators,  $1/|q|^2$ for some momentum $q$, with delta functions $\delta(|q|^2)$---captures precious analytic information about the integral. 
In particular, it is central in the ``method of differential equations'' (see e.g.~\cite{Henn:2014qga} for a review)
to obtain the so-called ``canonical form''~\cite{henn2013multiloop}, which then allows for a systematic solution in terms of special functions.

In our presentation here, we take a more general approach to derive differential operators which annihilate the integral.
Instead of using conformal symmetry, we construct the Picard--Fuchs operators~\cite{LV22,Muller-Stach:2012tgj}.
We do this by following a path which is the reverse of what we discussed in \Cref{sec:intPfaffian}.
First, we construct a Pfaffian system of PDEs for the ladder integral and other three integrals with the same propagators but different numerators, by following the method of differential equations.
We generate the required integration-by-parts identities with \textsc{LiteRed}~\cite{Lee:2012cn,Lee:2013mka} and solve them with \textsc{FiniteFlow}~\cite{Peraro:2019svx}.
We then differentiate the PDEs and use Gaussian elimination to obtain operators which annihilate~$J^{\text{ladder}}_{d}$.
The resulting differential operators are
\begin{align}
\begin{split}
P_1 & \, = \, 2 y_2 (y_2 - 7 y_3) (y_3 - 2 y_2) \partial_{y_2} \partial_{y_3}  + y_2 (y_2 - 7 y_3) (2 y_2 - y_3) \partial_{y_2}^2 \\
& \ \phantom{=} \ + y_2 (y_2 - 7 y_3) (2 y_2 - 
y_3) \partial_{y_3}^2  -6 y_2^2 (2 y_2 - 3 y_3) y_3 \partial_{y_2} \partial_{y_3}^2 \\
& \ \phantom{=} \ -2 y_2^2 (2 y_2 - 3 y_3) y_3 \partial_{y_2}^3 + 2 y_2^2 (2 y_2 - 3 y_3) y_3 \partial_{y_3}^3  + 6 y_2^2 (2 y_2 - 3 y_3) y_3 \partial_{y_2}^2 \partial_{y_3}\\
& \ \phantom{=} \ -4 y_2^3 y_3^2 \partial_{y_2}^3 \partial_{y_3} -4 y_2^3 y_3^2 \partial_{y_2} \partial_{y_3}^3 + y_2^3 y_3^2 \partial_{y_2}^4 + y_2^3 y_3^2 \partial_{y_3}^4  + 6 y_2^3 y_3^2 \partial_{y_2}^2 \partial_{y_3}^2\\
& \ \phantom{=} \ - (3 y_2^2 - 8 y_2 y_3 + y_3^2) \partial_{y_3} + (3 y_2^2 - 8 y_2 y_3 + y_3^2) \partial_{y_2} \,, \\
P_2 & \,=\, y_3^2 (y_2 + y_3) \partial_{y_3}^4 
    + 2 y_3 (2 y_2 + 3 y_3) \partial_{y_3}^3 + (2 y_2 + 7 y_3) \partial_{y_3}^2 +  \partial_{y_3}  \,.
\end{split}
\end{align}
Note that both operators admit the constant function as a solution.
This results from the chosen normalization of the ladder integral.

\begin{remark}
The choice of coordinates is highly relevant to the computations---a good choice can significantly speed up the Gröbner basis computations.
\end{remark}
Indeed, the variables defined in~\Cref{eq:ydef} are designed to simplify the solution functions.
We exploit the knowledge of the singular locus and choose the variables so that the polynomial defining the singular locus simplifies to $y_2 y_3 (y_2+y_3)$. 

A Gröbner basis $\{G_1,G_2\}$ of $\langle P_1,P_2\rangle$ in the rational Weyl algebra $R_2=\CC(y_2,y_3)\langle \partial_{y_2},\partial_{y_3}\rangle$ can be computed conveniently with the the command {\tt OreGroebnerBasis} in the {\tt Mathematica} package {\tt HolonomicFunctions}~\cite{HolFun}. The operators are given by
\begin{align}\begin{split}
 G_1 &\, = \, y_2 \partial_{y_2} + y_3 \partial_{y_3} \,, \\
 G_2& \, = \, y_3^2 (y_2 +y_3) \partial_{y_3}^4+ y_3 (4 y_2 +6 y_3) \partial_{y_3}^3+(2 y_2+7 y_3) \partial_{y_3}^2+ \partial_{y_3} \,.
\end{split}\end{align}
The $D_2$-ideal $I=\langle G_1,G_2\rangle$ as well as its Weyl closure $W(I)=R_2I\cap D_2$ are holonomic of rank~$4$. Since holomorphic functions in $y_2$ and $y_3$ form a torsion-free $\CC[y_2,y_3]$-module, a $D$-ideal and its Weyl closure have the same holomorphic functions as solutions. The Weyl closure of $I$ is generated by the three differential~operators
\begin{align}
\begin{aligned}
W_1& \,=\, y_2 \partial_{y_2}+ y_3 \partial_{y_3} \,, \\
W_2& \,=\, y_3^2 \partial_{y_2} \partial_{y_3}^3-y_3^2 \partial_{y_3}^4+3 y_3 \partial_{y_2} \partial_{y_3}^2-4 y_3 \partial_{y_3}^3+\partial_{y_2} \partial_{y_3}-2 \partial_{y_3}^2 \,, \\
W_3&  \,=\, y_2 y_3^2 \partial_{y_3}^4+y_3^3 \partial_{y_3}^4+4 y_2 y_3 \partial_{y_3}^3+6 y_3^2 \partial_{y_3}^3+2 y_2 \partial_{y_3}^2+7 y_3 \partial_{y_3}^2+\partial_{y_3} \,.
\end{aligned}
\end{align}
They may be obtained with the {\sc Singular} code below.
{\small
\begin{verbatim}
LIB "dmod.lib";
ring r = 0,(y2,y3,Dy2,Dy3),dp; def D2 = Weyl(r); setring D2;
poly G1 = y2*Dy2 + y3*Dy3;
poly G2 = (y2*y3^2+y3^3)*Dy3^4+(4*y2*y3+6*y3^2)*Dy3^3+(2*y2+7*y3)*Dy3^2+Dy3;
ideal I = G1,G2;
def WI = WeylClosure(I); WI;
\end{verbatim}
}

Note that, for the given generators, the $D$-ideal $\langle P_1,P_2 \rangle$ is already too complicated for {\sc Singular} to compute a Gröbner basis within a couple of minutes. 
But, one can check that $P_2\in W(I)\setminus I$. Indeed, $I \subsetneq W(I)=W(\langle P_1,P_2\rangle)$ as $D$-ideals. Their singular loci are
\begin{align}
 \Sing(I) \,=\, \Sing(W(I)) \,=\,   V(y_2y_3(y_2+y_3))  \, . 
\end{align}

The operator $G_1$ encodes that the solution functions $f(y_2,y_3)$ of the $D_2$-ideal~$I$ are homogeneous of degree~$0$, and therefore $f(y_2,y_3)=f(y_2/y_3,1)$. 
This is not obvious by looking  at the Feynman integral, and is therefore a powerful insight of the $D$-module~approach.

\begin{remark}
    It is possible to deduce that the system can be rewritten in one variable purely from the geometry of the Gr\"obner fan. The Gr\"obner fan has rays generated by the vectors $(1,1)$ and $(-1,-1)$, and is a one-dimensional fan in $\mathbb{R}^2,$ perpendicular to $(1,-1).$ The only cones in the dual fan are $\RR_{\geq 0}\cdot(1,-1)$ and $\RR_{\geq 0}\cdot(-1,1)$. Thus, the SST algorithm tells us that there are two families of series expansions in $y_2$ and $y_3$, with exponents only in $\NN \cdot (1,-1)$ and $\NN \cdot (-1,1),$ respectively.  From this, we see that we can rewrite these series expansions in the single variable $y = y_2y_3^{-1}.$
\end{remark}

We hence reduce the system to a single variable, namely the ratio $y \coloneqq y_2/y_3$. 
To do so, we change variables from $(y_2,y_3)$ to $(y,z)$ via 
\begin{align}
y \,=\, {y_2}y_3^{-1}, \quad z\,=\,y_3 \,.
\end{align}
Since the $D_2$-module in question is $D/I$, one can identify $z\partial_z =y_2\partial_{y_2}+y_3\partial_{y_3}\equiv 0$, so that we can disregard $G_1$ (keeping the homogeneity of its solutions in mind).
After this change of variables and multiplying by $z/y$, the operator $G_2$ becomes
\begin{align}
\frac{z}{y} \, G_2^y \,=\, y^3 (y+1)\partial_y^4+y^2(6+8y)\partial_y^3
+y (7+14y)\partial_y^2+(1+4y)\partial_y \, .
\end{align}
Denoting $\theta:=y\partial_y$ the Euler operator in $y$, $G_2^y$ becomes
\begin{equation}\label{eq:gop}
    G_2^y \,=\, \frac{1}{z}\left[ \theta^4 \,+\, y \, \theta^2\left(\theta+1\right)^2 \right] .
\end{equation}

For solving the associated differential equation, we can ignore the pre-factor of~$z^{-1}$ and hence are in the ODE case, with a single variable~$y$. 
 
Since the system is in a single variable, we will only need to distinguish between $w$ being positive or negative, and can hence restrict to the weight being $w=\pm  1$. The starting monomials for $w=1$ can be obtained from the initial ideal $\init_{(-1,1)}(G_2^y)=\langle \theta^4\rangle.$ They are 
\begin{align}
    1, \, \log(y),\,  \log(y)^2,\, \log(y)^3.
\end{align}
As in \Cref{eg:hypergeo}, we encode the coefficients of our Puiseux series in a  vector $c_p$ of length $4$ specifying the coefficients of $y^p, y^p \log(y), y^p \log(y)^2,$ and $y^p\log(y)^3.$ As before, we can encode the action of $\theta$ on the series as multiplication of the coefficient vector $c_p$ by the matrix
\begin{align}
    T_p \, = \, \begin{bmatrix}
    p + a & 1 & 0 & 0 \\
    0 & p + a & 2 & 0 \\
    0 & 0 & p + a & 3\\
    0 & 0 & 0 & p + a
\end{bmatrix} .
\end{align}
Here, $a$ is the coefficient of $y$ in the starting monomial, and we omit it from the index  when clear from context. Let $H_p = T_p^4$ and $F_p = -T_p^2(T_p  + \operatorname{Id}_4)^2.$ Then, functions annihilated by $G^y_2$ satisfy the recurrence relation
\begin{equation}
    H_p\cdot c_p \, = \, -F_{p-1}\cdot c_{p-1} \,.
\end{equation}
The first two starting monomials are themselves solutions. Let us start with $\log(y)^2,$ or rather the vector $c_0 = (0,0,1,0).$ Then $c_1 = (-2,0,0,0).$ Let $a_p$ be the coefficient of 
$y^p$. The upper left entry of the matrix $-H_p^{-1}F_{p-1}$ gives us the recurrence
\begin{equation}
    a_p \,=\, -\frac{(p-1)^2}{p^2} \, a_{p-1} \,,
\end{equation}
with $a_1 = -2.$ 
Next, let us start with $\log(y)^3,$ or rather the vector $c_0 = (0,0,0,1).$ Then $c_1 = (12,-6,0,0).$ Let $a_p$ be the coefficient of $y^p$, and $b_p$ be the coefficient of $y^p\log(y).$ Examining entries of $-H_p^{-1}F_{p-1}$ gives the recurrence
\begin{align}
    a_p  \,=\, -\frac{(p-1)^2}{p^2} \, a_{p-1}-\frac{2(p-1)}{p^3} \, b_{p-1} \, , \qquad
    b_p \,=\, -\frac{(p-1)^2}{p^2}\,c_{p-1} \,,
\end{align}
with $a_1 = 12 $ and $b_1 = -6.$ Solving the recurrences gives us the four closed-form solutions: 
\begin{align}\label{eq:series_w=1}
\begin{split}
    f_1(y) & \, = \, 1 \,, \qquad \quad 
    f_2(y)  \, = \, \log(y) \,, \\
    f_3(y) & \,=\, \log(y)^2 + 2\sum_{p=1}^{\infty} \frac{(-y)^p}{p^2}
    \,=\, \log(y)^2 + 2\,\text{Li}_2(-y) \, , \\
    f_4(y) & \,=\, \log(y)^3 -12\sum_{p \,=\, 1}^{\infty}\frac{(-y)^p}{p^3} + 6 \log(y)\sum_{p\,=\, 1}^{\infty} \frac{(-y)^p }{p^2} \\ & \,=\, \log(y)^3 -12 \,\text{Li}_3(-y)+6\log(y)\text{Li}_2(-y) \, . 
\end{split}\end{align}
One can easily verify that these are solutions to the operator $G_2^y$. 

\smallskip

Similarly, for $w = -1$ we obtain the exponents $\{0,-1\}$ and starting monomials
\begin{align}
    1, \, \log(y), \, y^{-1}, \,y^{-1}\log(y) \, .
    \end{align}
We use the recursion $F_p\cdot c_p = -H_{p+1}\cdot c_{p+1}.$ Here, there is a slight complication: with starting monomial $\log(y),$ the matrix $F_p$ has a singularity at $a = 0$ and $p=-1$. However, we use the definition of canonical series to impose that neither of the terms $y^{-1}, \, y^{-1} \log(y)$ appears in the series expansions of the remaining solutions. Explicitly, we set the vector $c_{-1}$ equal to~zero. As before, expressing the recurrences in closed form gives us the Laurent series
\begin{align}\begin{split}
g_1(y) & \,=\, 1 \,, \qquad \quad
g_2(y)  \,=\, \log(y) \,, \\
g_3(y) & \,=\, y^{-1} \,+\, y^{-1} \sum_{p \,=\, -1}^{-\infty}\frac{(-y)^p}{(p-1)^2} \,, \\
g_4(y) & \,=\, y^{-1}\log(y) \,-\, 2y^{-1}\sum_{p \,=\, -1}^{-\infty}\frac{p\cdot (-y)^p}{(p-1)^3}  \, +\, y^{-1}\log(y)\sum_{p \,=\, -1}^{-\infty}\frac{(-y)^p}{(p-1)^2} \, .
\end{split}\end{align}
These correspond to different series expansions of the di- and tri-logarithm in~\eqref{eq:series_w=1}.


\addcontentsline{toc}{section}{References}

\begin{thebibliography}{10}

\bibitem{AFST22}
D.~Agostini, C.~Fevola, A.-L.~Sattelberger, and S.~Telen.
\newblock {Vector Spaces of Generalized Euler Integrals}.
\newblock Preprint \href{https://arxiv.org/abs/2208.08967}{arXiv:2208.08967}. To appear in {\em Commun.\ Number Theory Phys.}, 2024.

\bibitem{solvePDE}
R.~Ait El~Manssour, M.~Härkönen, and B.~Sturmfels.
\newblock Linear {PDE} with constant coefficients.
\newblock {\em Glasg.\ Math. J.}, 65(S1):2--27, 2023.

\bibitem{Ananthanarayan:2022ntm}
B.~Ananthanarayan, S.~Banik, S.~Bera, and S.~Datta.
\newblock {FeynGKZ: a Mathematica package for solving Feynman integrals using
  GKZ hypergeometric systems}.
\newblock {\em Comput.\ Phys. Commun.}, 278(108699), 2023.

\bibitem{Anastasiou:1999ui}
C.~Anastasiou, E.~W.~N.~Glover and C.~Oleari.
\newblock {Scalar one loop integrals using the negative dimension approach}.
\newblock {\em Nucl.\ Phys.\ B}, 572:307--360, 2000.

\bibitem{ABLMS}
D.~Andres, M.~Brickenstein, V.~Levandovskyy, J.~Mart\'{i}n-Morales, and H.~Sch{\"o}nemann.
\newblock Constructive {D}-module theory with {\sc Singular}.
\newblock {\em Math.\ Comput.\ Sci}, 4(2--3):359--383, 2010.

\bibitem{Argeri:2007up}
M.~Argeri and P.~Mastrolia.
\newblock {Feynman Diagrams and Differential Equations}.
\newblock {\em Int.\ J.\ Mod.\ Phys.}, A22:4375--4436, 2007.

\bibitem{Barnes:2010jp}
E.~Barnes, D.~Vaman, C.~Wu, and P.~Arnold.
\newblock {Real-time finite-temperature correlators from AdS/CFT}.
\newblock {\em Phys.\ Rev.\ D}, 82(2), 2010.

\bibitem{Bautista:2019qxj}
T.~Bautista and H.~Godazgar.
\newblock {Lorentzian CFT 3-point functions in momentum space}.
\newblock {\em J.\ High Energy Phys.}, 01(142), 2020.

\bibitem{FeynmanAnnihilator}
T.~Bitoun, C.~Bogner, R.~P. Klausen, and E.~Panzer.
\newblock Feynman integral relations from parametric annihilators.
\newblock {\em Lett.\ Math.\ Phys.}, 109:497--564, 2019.

\bibitem{Mardi}
T.~Boege, R.~Fritze, C.~Görgen, J.~Hanselman, D.~Iglezakis, L.~Kastner, T.~Koprucki, T.~Krause, C.~Lehrenfeld, S.~Polla, M.~Reidelbach, C.~Riedel,
  J.~Saak, B.~Schembera, K.~Tabelow, and M.~Weber.
\newblock {Research-data management planning in the German mathematical community}.
\newblock {\em Eur.\ Math.\ Soc.\ Mag.}, 130:40--47, 2023.

\bibitem{Boos:1987bg}
E.~E. Boos and A.~I. Davydychev.
\newblock {A Method of the Evaluation of the Vertex Type Feynman Integrals}.
\newblock {\em Moscow Univ.\ Phys.\ Bull.}, 42N3:6--10, 1987.

\bibitem{Boos:1990rg}
E.~E.~Boos and A.~I.~Davydychev.
\newblock {A method of evaluating massive Feynman integrals}.
\newblock {\em Theor.\ Math.\ Phys.}, 89:1052--1063, 1991.

\bibitem{Braun:2003rp}
V.~M. Braun, G.~P. Korchemsky, and D.~M\"uller.
\newblock {The Uses of Conformal Symmetry in QCD}.
\newblock {\em Prog.\ Part.\ Nucl.\ Phys.}, 51:311--398, 2003.

\bibitem{Brown:2013qva}
F.~Brown.
\newblock {Iterated integrals in quantum field theory}.
\newblock In {\em {6th Summer School on Geometric and Topological Methods for Quantum Field Theory}}, pages 188--240, 2013.

\bibitem{Bzowski:2020lip}
A.~Bzowski.
\newblock {TripleK: A Mathematica package for evaluating triple-K integrals and
  conformal correlation functions}.
\newblock {\em Comput.\ Phys.\ Commun.}, 258(107538), 2021.

\bibitem{BMSconf}
A.~Bzowski, P.~McFadden, and K.~Skenderis.
\newblock {Implications of conformal invariance in momentum space}.
\newblock {\em J.\ High Energy Phys.}, 03(111), 2014.

\bibitem{Bzowski:2015yxv}
A.~Bzowski, P.~McFadden, and K.~Skenderis.
\newblock {Evaluation of conformal integrals}.
\newblock {\em J.\ High Energy Phys.}, 02:68, 2016.

\bibitem{Canko:2021xmn}
D.~D.~Canko and N.~Syrrakos.
Planar three-loop master integrals for $2 \to 2$ processes with one external massive particle. {\em J.\ High Energy Phys.}, 04:134, 2022.

\bibitem{Caron-Huot:2014lda}
S.~Caron-Huot and J.~M.~Henn.
\newblock Iterative structure of finite loop integrals.
\newblock {\em J.\ High Energy Phys.}, 06(114), 2014.

\bibitem{Chavez:2012kn}
F.~Chavez and C.~Duhr.
\newblock {Three-mass triangle integrals and single-valued polylogarithms}.
\newblock {\em J.~High Energy Phys.}, 11(114), 2012.

\bibitem{MacaulayFeynman}
V.~Chestnov, F.~Gasparotto, M.~K. Mandal, P.~Mastrolia, S.-J.\ Matsubara-Heo, H.~J.\ Munch, and N.~Takayama.
\newblock Macaulay matrix for {F}eynman integrals: Linear relations and intersection numbers.
\newblock {\em J.\ High Energy Phys.}, 09(187), 2022.

\bibitem{Chetyrkin:1981qh}
K.~G.\ Chetyrkin and F.~V.\ Tkachov.
\newblock {Integration by Parts: The Algorithm to Calculate beta Functions in $4$
Loops}.
\newblock {\em Nucl.\ Phys.\ B}, 192:159--204, 1981.

\bibitem{CHZ23}
D.~Chicherin, J.~Henn, and S.~Zoia.
\newblock {Anomalous Ward identities for on-shell amplitudes at the conformal fixed point}.
\newblock {\em J.\ High Energy Phys.}, 06(110), 2023.

\bibitem{Coriano:2013jba}
C.~Coriano, L.~Delle~Rose, E.~Mottola, and M.~Serino.
\newblock {Solving the Conformal Constraints for Scalar Operators in Momentum
  Space and the Evaluation of Feynman's Master Integrals}.
\newblock {\em J.~High Energy Phys.}, 07(011), 2013.

\bibitem{Davydychev:1992xr}
A.~I.~Davydychev.
\newblock {Recursive algorithm of evaluating vertex type Feynman integrals}.
\newblock {\em J.~Phys.~A}, 25:5587--5596, 1992.

\bibitem{Singular}
W.~Decker, G.-M.~Greuel,  G.~Pfister, and H.~Sch{\"o}nemann. 
\newblock {\sc Singular} {4-3-0} --- {A} computer algebra system for polynomial computations.
\newblock {\href{https://www.singular.uni-kl.de}{https://www.singular.uni-kl.de}} (2022).

\bibitem{Cruz19}
L.~de~la Cruz.
\newblock Feynman integrals as {A}-hypergeometric functions.
\newblock {\em J.\ High Energy Phys.}, 123(2019), 2019.

\bibitem{DiFrancesco:1997nk}
P.~Di~Francesco, P.~Mathieu, and D.~Senechal.
\newblock {\em {Conformal Field Theory}}.
\newblock Graduate Texts in Contemporary Physics. Springer-Verlag, New York,
  1997.

\bibitem{Drummond:2010cz}
J.~M. Drummond, J.~M. Henn, and J.~Trnka.
\newblock {New differential equations for on-shell loop integrals}.
\newblock {\em J.\ High Energy Phys.}, 04(083), 2011.

\bibitem{Duhr:2014woa}
C.~Duhr.
\newblock {Mathematical aspects of scattering amplitudes}.
\newblock In {\em {Theoretical Advanced Study Institute in Elementary Particle Physics}: {Journeys Through the Precision Frontier: Amplitudes for
  Colliders}}, pages 419--476, 2015.

\bibitem{mathrepo}
C.~Fevola and C.~G\"{o}rgen.
\newblock The mathematical research-data repository {M}ath{R}epo.
\newblock {\em {C}omputeralgebra {R}undbrief}, 70:16--20, 2022.

\bibitem{Gillioz:2019lgs}
M.~Gillioz.
\newblock {Conformal 3-point functions and the Lorentzian OPE in momentum
  space}.
\newblock {\em Commun.\ Math.\ Phys.}, 379(1):227--259, 2020.

\bibitem{Goncharov:2010jf}
A.~B.\ Goncharov, M.~Spradlin, C.~Vergu, and A.~Volovich.
\newblock {Classical Polylogarithms for Amplitudes and Wilson Loops}.
\newblock {\em Phys.\ Rev.\  Lett.}, 105(151605), 2010.

\bibitem{M2}
D.~R.\ Grayson and M.~E.\ Stillman.
\newblock Macaulay2, a software system for research in algebraic geometry.
\newblock Available at \url{http://www.math.uiuc.edu/Macaulay2/}.

\bibitem{Plural}
G.-M.\ Greuel, V.\ Levandovskyy, A.\ Motsak, and H.\ Schönemann. 
\newblock {\sc Plural}. A {\sc Singular}
Subsystem for Computations with Non-Commutative Polynomial Algebras, 2024.

\bibitem{higherlog}
R.~M. Hain and R.~MacPherson.
\newblock Higher logarithms.
\newblock {\em Illinois J.\ Math.}, 34(2), 1990.

\bibitem{henn2013multiloop}
J.~M. Henn.
\newblock Multiloop integrals in dimensional regularization made simple.
\newblock {\em Phys.\ Rev.\  Lett.}, 110(25):251601, 2013.

\bibitem{Henn:2014qga}
J.~M.\ Henn.
\newblock {Lectures on differential equations for Feynman integrals}.
\newblock {\em J.\ Phys.\ A}, 48(153001), 2015.

\bibitem{Henn:2023vbd}
J.~M.~Henn, J.~Lim, and W.~J.~Torres Bobadilla,
{\rm First look at the evaluation of three-loop non-planar Feynman diagrams for Higgs plus jet production}. {\em J.~High Energy Phys.}  05(026), 2023.

\bibitem{HTT08}
R.~Hotta, K.~Takeuchi, and T.~Tanisaki.
\newblock {\em {$D$}-{M}odules, Perverse Sheaves, and Representation Theory}, volume 236 of {\em Progr.\ Math.}.
\newblock Birkh\"{a}user Boston, 2008.

\bibitem{Isono:2019ihz}
H.~Isono, T.~Noumi, and T.~Takeuchi.
\newblock {Momentum space conformal three-point functions of conserved currents and a general spinning operator}.
\newblock {\em J.\ High Energy Phys.}, 05:057, 2019.

\bibitem{MellinBarnes}
M.~Yu.~Kalmykov and B.~A.~Kniehl.
\newblock Mellin--Barnes representations of Feynman diagrams, linear systems of differential equations, and polynomial solutions.
\newblock {\em Phys.\ Lett.\  B},
714(1):103--109, 2012.

\bibitem{KK1976}
M.~Kashiwara and T.~Kawai.
\newblock {Holonomic Systems of Linear Differential Equations and Feynman
  Integrals}.
\newblock {\em Publ.\ Res.\ Inst.\ Math.\  Sci.}, 12(99):131--140, 1976.

\bibitem{HolFun}
C.~Koutschan.
\newblock {\em Advanced Applications of the Holonomic Systems Approach}. RISC, Johannes Kepler University, Linz. PhD Thesis. September 2009.

\bibitem{LV22}
P.~Lairez and P.~Vanhove.
\newblock {Algorithms for minimal Picard--Fuchs operators of Feynman integrals}.
\newblock {\em Lett.\ Math.\ Phys.}, 113(2), 2023.

\bibitem{Laporta:2000dsw}
S.~Laporta.
\newblock {High precision calculation of multiloop Feynman integrals by difference equations}.
\newblock {\em Int.\ J.\ Mod.\ Phys.\ A}, 15:5087--5159, 2000.

\bibitem{Lee:2013hzt}
R.~N.\ Lee and A.~A.\ Pomeransky.
\newblock {Critical points and number of master integrals}.
\newblock {\em J.\ High Energy Phys.}, 11(165), 2013.

\bibitem{Lee:2013mka}
R.~N.~Lee.
\newblock {LiteRed 1.4: a powerful tool for reduction of multiloop integrals}.
\newblock {\em J.\ Phys.\ Conf.\ Ser.}, 523(012059), 2014.

\bibitem{Lee:2012cn}
R.~N.~Lee.
\newblock {Presenting LiteRed: a tool for the Loop InTEgrals REDuction}.
\newblock Preprint \href{https://arxiv.org/abs/1212.2685}{arXiv:1212.2685}, 2012.

\bibitem{Lee:2014ioa}
R.~N. Lee.
\newblock {Reducing differential equations for multiloop master integrals}.
\newblock {\em J.\ High Energy Phys.}, 04:108, 2015.

\bibitem{Liu:2022chg}
X.~Liu and Y.~Q.~Ma.
\newblock{AMFlow: A Mathematica package for Feynman integrals computation via auxiliary mass flow}.
\newblock{\em Comput.\ Phys.\ Commun.}, 283(108565), 2023.

\bibitem{Muller-Stach:2012tgj}
S.~M\"uller-Stach, S.~Weinzierl, and R.~Zayadeh.
\newblock {Picard--Fuchs equations for Feynman integrals}.
\newblock {\em Commun.\ Math.\ Phys.}, 326:237--249, 2014.

\bibitem{Peraro:2019svx}
T.~Peraro.
{\rm FiniteFlow: multivariate functional reconstruction using finite fields and dataflow graphs},
{\em J.\ High Energy Phys.} 07(031), 2019.

\bibitem{Polyakov:1970xd}
A.~M.~Polyakov.
\newblock {Conformal symmetry of critical fluctuations}.
\newblock {\em JETP Lett.}, 12:381--383, 1970.

\bibitem{SST00}
M.~Saito, B.~Sturmfels, and N.~Takayama.
\newblock {\em {Gr\"{o}bner deformations of hypergeometric differential
  equations}}, volume~6 of {\em Algorithms and Computation in Mathematics}.
\newblock Springer, 2000.

\bibitem{SatStu19}
A.-L.~Sattelberger and B.~Sturmfels.
\newblock {{$D$}-Modules and Holonomic Functions}.
\newblock Preprint \href{https://arxiv.org/abs/1910.01395}{arXiv:1910.01395}, 2019. To appear in the volume {\em Varieties, polyhedra, computation} of {\em EMS Series of Congress Reports}.

\bibitem{TH22}
F.~Tellander and M.~Hellmer.
\newblock{Cohen--Macaulay Property of Feynman Integrals}.
\newblock{\em Commun.\ Math.\ Phys.}, 2022. \href{https://doi.org/10.1007/s00220-022-04569-6}{https://doi.org/10.1007/s00220-022-04569-6}

\bibitem{Tkachov:1981wb}
F.~V.\ Tkachov.
\newblock {A Theorem on Analytical Calculability of Four Loop Renormalization
  Group Functions}.
\newblock {\em Phys.\ Lett.\  B}, 100:65--68, 1981.

\bibitem{Usyukina:1992jd}
N.~I. Usyukina and A.~I. Davydychev.
\newblock {An approach to the evaluation of three and four point ladder
  diagrams}.
\newblock {\em Phys.\ Lett.\  B}, 298:363--370, 1993.

\bibitem{Usyukina:1994iw}
N.~I.\ Usyukina and A.~I.\ Davydychev.
\newblock {New results for two loop off-shell three point diagrams}.
\newblock {\em Phys.\ Lett.\  B}, 332:159--167, 1994.

\bibitem{Usyukina:1992wz}
N.~I. Usyukina and A.~I. Davydychev.
\newblock {Some exact results for two loop diagrams with three and four
  external lines}.
\newblock {\em Phys. Atom. Nucl.}, 56:1553--1557, 1993.

\bibitem{vdPS}
M.~van~der Put and M.~Singer.
\newblock {\em Galois Theory of Linear Differential Equations}, volume 328 of
  {\em Grundlehren der mathematischen Wissenschaften}.
\newblock Springer, 2003.

\bibitem{DiVita:2014pza}
S.~Di Vita, P.~Mastrolia, U.~Schubert, and V.~Yundin.
{\rm Three-loop master integrals for ladder-box diagrams with one massive leg}. {\em J.\ High Energy Phys.}, 09(148), 2014. 

\bibitem{Walther22}
U.~Walther.
\newblock {On {F}eynman graphs, matroids, and {GKZ}-systems}. 
\newblock {\em Lett. Math. Phys.}, 112(120), 2022. 

\bibitem{Wasow}
W.~Wasow.
\newblock {\em Asymptotic expansions for ordinary differential equations}.
\newblock Pure and Applied Mathematics, Vol.\ XIV.\ Interscience Publishers John Wiley \& Sons, Inc., New York-London-Sydney, 1965.

\bibitem{Zei90}
D.~Zeilberger.
\newblock A holonomic systems approach to special functions identities.
\newblock {\em J.\ Comput.\ Appl.\ Math.}, 32(3):321--368, 1990.

\bibitem{Zoia:2021zmb}
S.~Zoia.
\newblock {\em {Modern Analytic Methods for Computing Scattering Amplitudes: With Application to Two-Loop Five-Particle Processes}}.
\newblock Springer Theses. Springer Cham, 2022.

\end{thebibliography}
\end{document}